\documentclass[a4paper,10pt]{article}
 \usepackage[english]{babel}
\usepackage[utf8]{inputenc}  
\usepackage[T1]{fontenc}  
\usepackage{amsmath,amsfonts,amssymb,color,amsthm}   
\usepackage{pdfpages}

\usepackage{hyperref}
\usepackage{natbib}
\usepackage{rotfloat}

\usepackage[boxruled]{algorithm2e}

\usepackage[svgnames]{xcolor}

\newtheorem{proposition}{Proposition}

\newtheorem{remark}{Remark}

\usepackage[margin=1in]{geometry}
\usepackage{setspace}

\newcommand{\btheta}{\boldsymbol{\theta}}
\newcommand{\ba}{\mathbf{a}}
\newcommand{\bb}{\mathbf{b}}
\newcommand{\bc}{\mathbf{c}}
\newcommand{\bY}{\mathbf{Y}}
\newcommand{\bX}{\mathbf{X}}
\newcommand{\bbX}{\mathbb{X}}
\newcommand{\bO}{\boldsymbol{\omega}}
\newcommand{\bZ}{\mathbf{Z}}

\newcommand{\E}{\mathbb{E}}

\newcommand{\new}[1]{\textcolor{black}{#1}}
\newcommand{\nnew}[1]{\textcolor{black}{#1}}

\title{Post-processing multi-ensemble temperature and precipitation forecasts 
through an Exchangeable Gamma Normal model and its Tobit extension}
\author{ Marie Courbariaux$^1$
\and Pierre Barbillon$^1$\footnote{corresponding author: pierre.barbillon@agroparistech.fr} \and Luc Perreault$^2$ \and \'Eric Parent$^1$\\$^1$UMR MIA-Paris, AgroParisTech, INRA,\\ Universit\'e Paris-Saclay, 75005, Paris, France \\$^2$Hydro-Québec Research Institute, Varennes, Canada }

\begin{document}

\maketitle


\begin{abstract}

Meteorological ensemble members are a collection of scenarios for future weather issued by a meteorological center. Such ensembles nowadays form the main source of valuable information for probabilistic forecasting which aims at producing a predictive probability distribution of the quantity of interest instead of a single best guess point-wise estimate. Unfortunately, ensemble members cannot generally be considered as a sample from such a predictive probability distribution without a preliminary post-processing treatment to re-calibrate the ensemble. Two main families of post-processing methods, either competing such as the BMA or collaborative such as the EMOS, can be found in the literature. This paper proposes a mixed effect model belonging to the collaborative family. The structure of the model is formally justified by Bruno de Finetti's representation theorem which shows how to construct operational statistical models of ensemble based on judgments of invariance under the relabeling of the members. Its interesting specificities are as follows:
  1) exchangeability contributes to parsimony, with an interpretation of the latent pivot of the ensemble in terms of a statistical synthesis of the essential meteorological features of the ensemble members,
  2) a multi-ensemble implementation is straightforward, allowing to take advantage of various information so as to increase the sharpness of the forecasting procedure.
 Focus is cast onto Normal statistical structures, first with a direct application for temperatures, then with its very convenient Tobit extension for precipitation. 
 Inference is performed by Expectation Maximization (EM) algorithms with both steps leading to explicit analytic
 expressions in the Gaussian temperature case and recourse is made to stochastic conditional simulations in the zero-inflated precipitation case. 
 After checking its good behavior on artificial data, the proposed post-processing technique is applied to temperature and precipitation ensemble forecasts produced for lead times from 1 to 9 days over five river basins managed by Hydro-Québec, which ranks among the world’s largest electric companies. These ensemble forecasts, provided by three meteorological global forecast centres (Canadian, US and European), were extracted from the THORPEX Interactive Grand Global Ensemble (TIGGE) database. The results indicate that post-processed ensembles are calibrated and generally sharper than the raw ensembles for the five watersheds under study.

\end{abstract}


\section{Introduction}

Rather than a single scenario (\textit {i.e. a deterministic prediction}),
meteorological services are now delivering a full collection of scenarios (that are called \textit {the members of the ensemble}) as an attempt to depict their knowledge as well as their uncertainty about the future state of the atmosphere.
The members of the ensemble are obtained by introducing perturbations into the initial conditions or the parametrization of a numerical weather prediction model (NWP), or by using several NWP models.
As an example, the European Center for Medium-range Weather Forecasts (ECMWF) generates the 50 members of the ECMWF-EPS (Ensemble Prediction System)
by launching every day 50 runs of their NWP model with different initial conditions as described in \citep {buizza2008potential}.

One may expect that the ensemble members $X_{k,t,h}, k=1\ldots K$, \new{issued at time $t-h$}, will behave as a $K$-sample from the probabilistic forecast of the variable $\new{Y_{t}}$  
to be predicted $h$ days ahead (e.g. temperature, wind, precipitation and so on).
This hypothesis of \textit{calibration} or \textit{reliability} for meteorological ensembles is sometimes assumed by weather forecasts users such as hydropower companies.  
However this is generally not tenable for most ensemble prediction systems. In fact, as pointed out by, among others, \cite{hamill1997verification,bougeault2010thorpex, taillardat2016calibrated} and \cite {perreault2017posttraitement}, ensemble weather forecasts are subject to bias and tend to be underdispersive. 
\new{In the following, the index $h$ will be omitted and we will write $X_{k,t}$ instead of $X_{k,t,h}$ since the lead time $h$ is considered  fixed.}

As an illustration, the three top panels of Figure~\ref{Fig-HistRgM2VLCens} show the rank histograms for 3-days ahead precipitation ensemble forecasts produced in 2014 for Manic 2 watershed by the European (ECMWF-EPS), 
Canadian (CMC\footnote{CMC: Canadian Meteorological Center.}-EPS) and US (NCEP-GEFS\footnote{NCEP: National Center for Environmental Prediction; GEFS: Global Ensemble Forecast System.}) 
meteorological services. The marked deviations from the uniform distribution of these histograms show that these ensembles are all far from being calibrated: the raw ensembles are biased and under-dispersed \citep{hamill2001interpretation}. It is therefore necessary to carry out some form of statistical post-processing on the ensemble members in order to simulate a new $K$-sample, which will then be well calibrated. Several methods of statistical post-processing have been proposed to improve ensemble weather forecasts. Two main strategies can be found in the literature: the members will either compete or collaborate.

\begin{figure}[!h]
\centering 

 \includegraphics[scale=0.7]{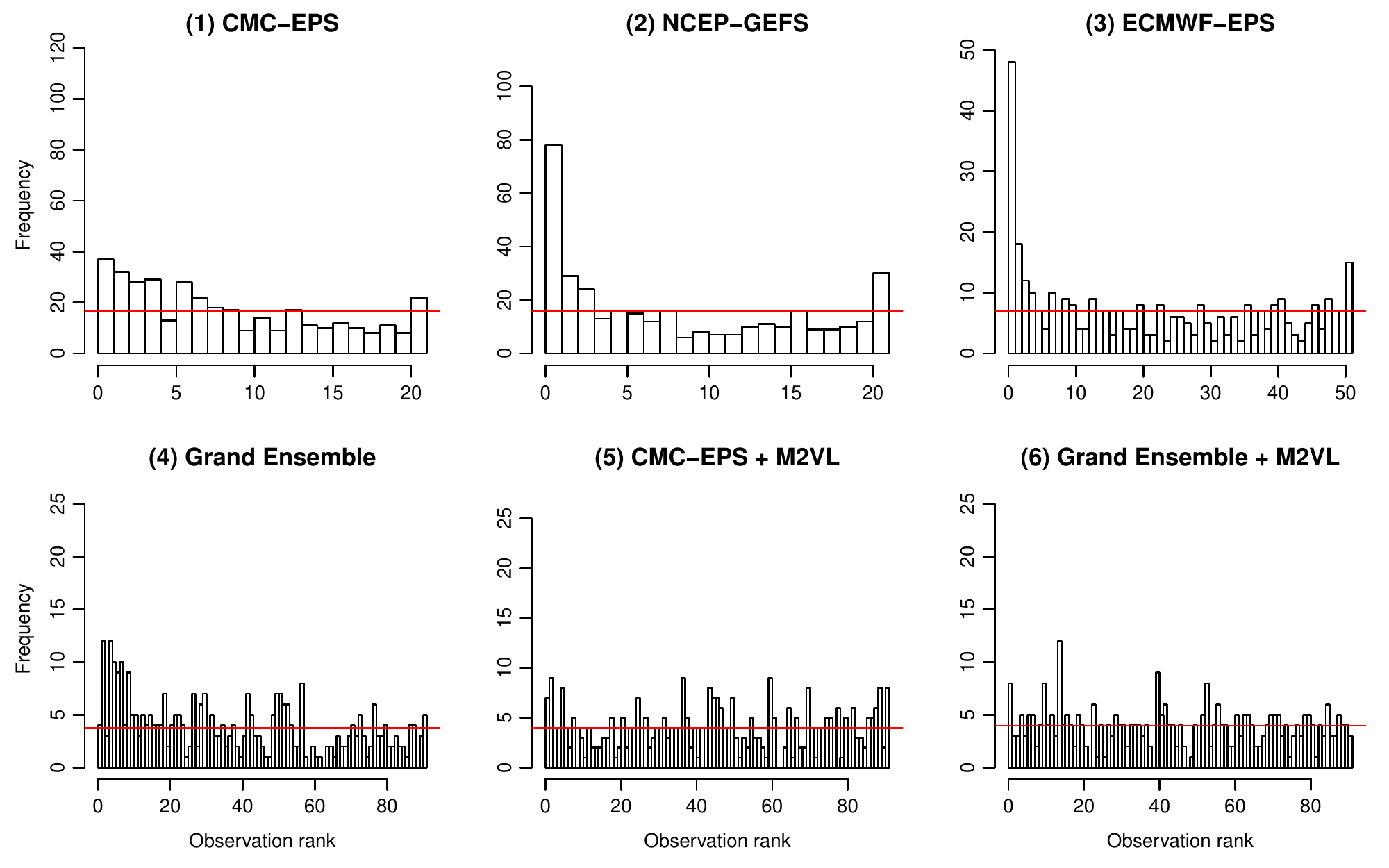}
  \caption{
Rank histograms for 3-days lead time ensemble precipitation forecasts produced daily in 2014 for Manic 2 watershed: (1) CMC-EPS raw ensemble (2) NCEP-GEFS raw ensemble (3) ECMWF-EPS raw ensemble (4) large raw ensemble gathering CMC-EPS, NCEP-GEFS and ECMWF-EPS members (5) post-processed CMC-EPS (6) post-processed large ensemble.}
 \label{Fig-HistRgM2VLCens}
\end{figure}

The first family considers that each member $k$ in itself can be the potential `\textit{truth}' $X_{k,t}$ for the quantity $Y_t$ to forecast. As if the $K$ members of the ensemble were competing, the idea is to identify statistically some \textit{best} member $k^*$ and try to generate ensemble forecasts in the vicinity of $X_{k^*,t}$. To provide an adjusted predictive distribution, one generally relies on weighted smoothing kernels over the ensemble members.
The statistical structure underpinning most methods from this family is as follows: for any $t$ and $k$,
\begin{align}
\label{Eq-BMA}
\left(Y_t|\bX_{t},k^\star \right) & \sim\mathcal{G}\left(X_{k^\star,t},\boldsymbol{\theta}_{k^\star}\right)\\
Pr(k^\star=k)  & =\pi_k \nonumber \; 
\end{align}
where $\bX_t=\{X_{1,t},\ldots,X_{K,t}\}$, $\sum_{k=1}^{K}\pi_k=1$, $\mathcal{G}$ is a probability distribution function (pdf) to be chosen,
$\boldsymbol{\theta}_k$ and $\pi_k$, for $k \in \left\{1,\ldots ,K\right\}$, are parameters to be estimated and $k^\star$ is the index of the (unknown) best member. As it is usually done, post-processing is applied independently for each lead time, the temporal dependency being reconstructed using empirical copulas.
The most famous method in the competing family is the Bayesian Model Averaging (BMA)  proposed by \cite{raftery2005using}. They considered the Gaussian case, taking $\mathcal{G}$
under the following form: for any $t$ and $k$,
\begin{equation*}
\mathcal{G}\left(X_{k,t},\boldsymbol{\theta}_{k}\right)=
\mathcal{N}\left(a_{k}X_{k,t}+b_{k},\sigma^2\right),
\end{equation*}
with $\boldsymbol{\theta}_{k}=\left(a_{k},b_{k},\sigma\right)$ 
for $k \in \left\{1,\ldots ,K\right\}$. 
The EM (Expectation-Maximization) algorithm is used for inference and the resulting predictive distribution in this case is a mixture of Gaussian distributions. Of course, this is a way of re-assembling the members to work together for prediction at the end, but this collaboration is performed with uneven weights, at least in principle. Consequently the BMA can also be understood as a post-processing \textit{dressing} method with a Gaussian kernel. This method is not without drawback when the ensemble is over-dispersive as shown by a simulated data experiment conducted by \cite{raftery2005using}. However this case rarely happens in practice since ensemble weather prediction systems tend to produce overconfident forecasts: the resulting ensemble members are generally less dispersed than they should be.
Many extensions of the BMA have been proposed, changing the method of inference or the distribution of the ensemble. As an example for precipitation \cite{sloughter2007probabilistic} proposed a power transformation ($1/3$) and the following Bernoulli-Gamma model for $\mathcal{G}$ so as to deal with the case of no rain: for any $t$,
\begin{center}
$\begin{cases}
\text{logit}\left(  Pr(Y_t^{\frac{1}{3}}=0|k^\star=k,X_{k,t})\right)   & =d_{0}+d_{1}%
X_{k,t}^{\frac{1}{3}}+d_{2}\mathbb{I}_{\left(  X_{k,t}=0\right)  }\\
\left(  Y_t^{\frac{1}{3}}|k^\star=k,X_{k,t},Y_t>0\right)   & \sim\Gamma\left(  \alpha\left(
X_{k,t}\right)  ,\beta\left(  X_{k,t}\right)  \right)  \\
\alpha\left(  X_{k,t}\right)  \text{ and }\beta\left(  X_{k,t}\right)   & \text{s.
t. }\frac{\alpha\left(  X_{k,t}\right)  }{\beta\left(  X_{k,t}\right)  }%
=b_{0}+b_{1}X_{k,t}^{\frac{1}{3}}\\
& \text{and  }\new{\frac{\alpha\left(  X_{k,t}\right)  }{\beta\left(  X_{k,t}\right)^2
}=c_{0}+c_{1}X_{k,t}}\\%
Pr(k^\star=k)  & =\frac{1}{K}  \; \forall k  \nonumber
\end{cases}
$
\end{center}
 where \new{$\mathbf{d}=(d_0,d_1,d_2)$, $\mathbf{b}=(b_0,b_1)$ and $\mathbf{c}=(c_0,c_1)$} are vectors for parameters to be estimated and $\Gamma\left(\alpha,\beta\right)$ stands for the gamma distribution with shape $\alpha$ and rate $\beta$. 
 
The second family of post-processing techniques for meteorological ensembles considers that members are not alternative individualized proposals for the quantity to be predicted, but rather a collection of scenarios sharing common traits that are identified as summaries of the future state of the system to be predicted. The predictive distribution is to be based on these shared features.
The Bayesian Processor of Output (BPO) \citep{krzysztofowicz2006bayesian} suggests to consider the joint distribution $(\bX_t,Y_t)$ (after a normal quantile transform) in order to derive the conditional distribution of $Y_t$ given $\bX_t$. Without any additional hypothesis on the form of covariance between members, this model lacks parsimony. To our view, this is a major drawback when considering the limited sample size of historical data available to learn about the distribution of the ensemble $\bX_t$.
A first step towards parsimony would assume that the weighted mean of the ensemble is the sufficient statistic for some parametric modeling of the predictive distribution (generally, the ordinary mean is used). One can then imagine treating the ensemble mean as a deterministic forecast and consider a bivariate Gaussian model to link this statistic with the quantity to forecast $Y_t$. 
The EMOS (Ensemble Model Output Statistics) method proposed by \cite{gneiting2005calibrated} additionally assumes that the ensemble dispersion also informs on the variability of the future meteorological state. The proposed predictive parametric model is as follows: for any $t$,
\begin{equation}
\left(Y_t|\bX_t\right)=a+\mathbf{b}^T \bX_t+\sqrt{c+d S_{\bX_t}^2} \varepsilon_t, \quad
\varepsilon_t\sim\mathcal{N}\left(0,1\right),
\label{eq:EMOS}
\end{equation}
where $a$, $c>0$, $d>0$ and $\mathbf{b}$ are parameters to be estimated, and $S_{\bX}$ denotes the standard deviation of the ensemble $\bX$
.
Conversely to the BMA, EMOS cannot provide a multimodal predictive distribution. One can also encounter post-processing methods for \new{ensemble precipitation forecasts} in the collaborative family, such as the EMOS-like model proposed by \cite{scheuerer2014probabilistic}: for any $t$ and $k$,
\begin{center}
$\begin{cases}
\left(Z_t|\bX_t \right) & \sim GEV\left(\mu_{\bX_t},\gamma+\kappa MD_{\bX_t},\xi\right)\\
\mu_{\bX_t}\textrm{ s.t. }\mathbb{E}\left(Z_t|\bX_t\right) &
=\alpha+\beta\bar{\bX_t}+\frac{1}{K}\sum_{k=1}^{k=K}\mathbb{I}_{X_{k,t}=0}\\
MD_{\bX_t} & =\frac{1}{K^{2}}\sum_{k,k'}\left|X_{k,t}-X_{k^\prime,t}\right|\\
\left(Y_t|Z_t\right) & =Z_t\mathbb{I}_{Z_t\geq0}
\end{cases}$
\par\end{center}
where $\alpha$, $\beta$, $\gamma$, $\kappa$ et $\xi$ are parameters to be estimated,
$MD_{\bX}$ is a measure of dispersion for the ensemble $\bX$, $Z_t$ is the latent non censored variable for the precipitation $Y_t$
and $GEV(\mu,\sigma,\xi)$ is the generalized extreme value distribution with location parameter
$\mu$, dispersion parameter $\sigma>0$ and shape parameter $\xi$.
\\
Finally, although they do not provide an explicit predictive function but allow for estimation of the quantiles, methods closely related to quantile regression and nonparametric regression can also be considered to belong to the family of collaborative post-processing techniques.  For instance, following the \new{logistic} regression works of \cite{wilks2009extending}, \cite{ben2013calibrated} proposes to add an interaction term in this post-processing technique and \cite{messner2014heteroscedastic} developed a heteroscedastic version in order to use the information contained in the dispersion of the members. \cite{taillardat2016calibrated} combined quantile regression with random forests (\cite{meinshausen2006quantile}), for the post-processing of ensemble temperatures and wind speed predictions.

The purpose of this paper is to formalize and develop a new collaborative post-processing technique in the vein of EMOS but allowing to incorporate the information conveyed by multiple ensembles into the analysis. 
\new{Such a collaborative post-processing approach avoids the main shortcomings of dressing methods (a poor adaptation to cases where the ensembles are \new{over}-dispersive) and of the non-parametric methods, which are not originally intended to issue complete predictive distributions.} 

\new{Multi-model ensemble prediction is increasingly used since number of studies have demonstrated that these forecasts have higher prediction skill than that of an individual model (\citet{tebaldi2007use}). Such  “grand ensemble” are usually considered for long term meteorological forecasting, namely for seasonal forecasts (\citet{khajehei2018effective}), and for hydrological forecasting. Even though using raw multiple sources of meteorological forecasts may improve reliability, they still lack of calibration. Except for the BMA approach which belongs to the competing family (\citet{fraley2010calibrating}), no parametric method formally addresses multi-ensemble forecasts post-processing. Especially, 
to the best of our knowledge, no collaborative statistical post-processing model using latent variables, such as the EMOS extension presented herein, have been proposed to explicitly calibrate multiple ensembles forecasts. This is confirmed by \citet{li2017review} in their recent review on statistical post-processing methods for hydrometeorological ensemble forecasting.}

In this paper, we focus on medium range forecasting of daily temperature and precipitation that are of utmost interest for hydropower companies, since these two quantities are the main inputs of the rainfall-runoff model used to produce streamflow predictions  
(\citet{guay2018hsami1,courbariaux2017water,garcon1996prevision}). 
Strategic decisions are taken daily on the basis of these forecasts, namely to prevent flooding damages and to avoid operating losses. The availability of reliable temperature and precipitation probabilistic forecasts in this context is therefore a crucial issue.

Being explicitly underpinned by the theory of exchangeability, the proposed technique will benefit from a statistical property of symmetry (invariance by relabeling)  that can be naturally expected from efficient ensemble simulators. 
The statistical framework of exchangeability is recalled in Section \ref{sec:Exchangeability}. Its
main interest lies in providing a theoretical justification for the representation of the ensemble as a sample 
from a mixed effect model. Figuring out the conditioning latent process as some salient configuration of the ensemble simulator for the day to predict can help to understand the mixed model in meteorological terms.
Section \ref{sec:Gaussian-case} develops exchangeable applications for the Gaussian case with a noticeable multidimensional
mixed model that might improve forecasting sharpness by taking into account the multiple sources of information conveyed by different
ensemble simulators or by forecasts issued from various meteorological experts. This first model is suitable for temperature ensemble forecasts. Section \ref{sec:Tobit-case} 
is devoted to the Tobit extension of the exchangeable Gaussian model to deal with the zero inflation of the precipitation distribution corresponding to days with no rain. A simulation study reported in Section~\ref{sec:Numerical-experiment} allows us to check that the stochastic expectation maximization algorithm proposed for inference works properly on artificial data. Section \ref{sec:Applications} presents results for temperature and precipitation forecasts based on combining the European, Canadian and US ensembles for five watersheds in Quebec, Canada. The results indicate that post-processed ensembles are much better calibrated and generally sharper than the raw ensembles for the watersheds under study. Section \ref{sec:Conclusion-and-discussion} provides a summary and discusses perspectives of future research.

\section{Exchangeability}
\label{sec:Exchangeability}
Suppose that the labeling of the ensemble members has no impact on our prior beliefs: their joint distribution will remain invariant to relabeling
the members. Such an hypothesis of exchangeability seems plausible: at least this would ideally represent the desiderata for the proficient meteorologist willing to tune his earth model initial-condition perturbations such that the ensemble members do not exhibit systematically persisting figures over time. \cite{fraley2010calibrating} rely on exchangeability to assume in Eq~\eqref{Eq-BMA} $\pi_k=\frac{1}{K}$ and $\boldsymbol{\theta}_k=\theta \quad \forall k\in \left\{1,\ldots ,K\right\}$, and so do other authors with respect their own favorite post-processing technique. But the concept of exchangeability should be farther exploited since it provides formal means to construct operational statistical models of ensemble based strictly on judgments of invariance under the relabeling of the members.

Consider a $K$-sample  $(X_1,X_2,\ldots,X_k,\ldots,X_K)$, such that there exists a random variable $Z$ (with pdf $g(.)$ ) allowing to write the joint distribution of the $X_k,\,k\in\{1,\ldots,K\}$ as a mixed effect model:
\begin{equation}
 f(x_{1:K})=\int_z  g(z) \left(\prod_{k=1}^{K} f(X_{k}=x_k|Z=z)\right) dz.
 \label{eq:xchgt}
\end{equation}
Given the random effects $Z$, the $X_k,\,k\in\{1,\ldots,K\}$, are independent and identically distributed according to $f(x\vert Z)$.  The joint distribution of the variables $(X_1,X_2,\ldots,X_k,\ldots,X_K)$,  remains invariant under a relabeling permutation of the components of the mixture: they are exchangeable. Bruno de Finetti's representation theorem (\cite{de1931funzione,de1937prevision}) and the work from his followers \citep{hewitt1955symmetric} prove the difficult reciprocal: under technical conditions of regularity valid for a theoretically infinite sequence of exchangeable members, exchangeability means conditional independence and yields to Eq~\eqref{eq:xchgt}.
Note that exchangeability does permit marginal dependence between members; for example in the Gaussian case, members must have the same mean and the same variance, but they can be correlated with one another, as long as all correlations are equal (and positive).
Exchangeability is especially important for  modeling ensemble members by its realism as well as its parsimony. Moreover, there exists a very strong \textit{a priori} argument in favor of a structured model of the ensemble members  (the $X_{k}'s$ in eq~\ref{Eq-BMA}) around a latent single conditioning variable that would explain their dependencies. It is the very objective of the simulation method of the hydrometeorological model of the earth system to target, through a given ensemble of exchangeable simulations, the estimation of a latent meaningful "physical" variable for all members. Such a $Z$ (in Eq~\eqref{eq:xchgt}) can be interpreted as a statistical synthesis of this latent variable common to the members of the ensemble. This interpretation is furthermore often confined by the strength of the first component obtained through a preliminary PCA (Principal Component Analysis) of the ensemble.\\ 
One may 
object that some physical processes to generate ensembles do not produce exchangeable members: for instance, even and odd numbered members from the Canadian ensemble (CMC-EPS) for precipitations are obtained by different meteorological models\footnote{\url{http://collaboration.cmc.ec.gc.ca/cmc/ensemble/doc/info_geps_e.pdf"}} 
(see second panel of Fig~\ref{Fig-ExManic2-raw} ) and would not pass the tests for accepting exchangeability for the whole CMC-EPS.\\ 
\new{The rank statistics  
lead to a possible exchangeability test for ensembles. 
A necessary condition for exchangeability is 
indeed that each member of the ensemble occupies all possible ranks (i.e. $\{1,\cdots,K\}$) in a roughly equal way. We therefore check that the frequency of occupancy of a rank does not deviate too far from $\frac{1}{K}$ for each member (with a $\chi^2$ test).
}\\
In the case exchangeability is not an acceptable assumption for the whole ensemble, one can subset the ensemble into exchangeable parts to be dealt with as new ensembles.
\new{We thus get back to the case of exchangeability of members belonging to different ensembles.}
This will not prevent using the models hereafter since, in this paper, we develop a method allowing to incorporate the information conveyed by multiple ensembles into the statistical analysis.

\section{Model and inference in the Gaussian case}
\label{sec:Gaussian-case}

As in most post-processing techniques, we consider univariate models to obtain a calibrated marginal distribution for each site, each meteorological variable (here, the temperature), and for each lead time. \new{Indeed, the ensemble is generally assumed to be a sufficient summary statistics as far as prediction is concerned.} The spatial, temporal and inter-variable dependencies are recovered by using empirical copulas.

In this section, we consider the simple situation in which the variable to be predicted, $Y_t$ 
as well as the corresponding ensemble of predictors, \new{$ \mathbf {X}_t $ (produced a time $t-h$, $h$ being the considered lead-time)}, can be assumed jointly (and conditionally) Gaussian. 
This is notably the case for daily temperature observations over a season but it may also concern other continuous variables after a suitable normalizing transformation.

\subsection{Multi-ensemble Exchangeable Gamma Normal Model}
Let $E$ be the number of forecast sources (e.g. the ensembles from several meteorological centers) and $K_e$ the number of members within ensemble $e$.
When the forecasting system delivers a single forecast, for instance in the case of an expert issuing a deterministic forecast, we simply set $K_e=1$.
From now on, we also make the convention that $e=0$ will denote the variable to be predicted \new{ and we will sometimes conveniently write: $X_{0,t}=Y_t$, with $K_{\text{0}}=1$, as if it were a peculiar ensemble with a single member. This notational trick will be useful for the inference part, where it makes sense since past observations of the target provide information about the unknown of our model just as ensemble members do. Of course, when forecasting, predictand $Y_t$ and predictors $X_{e,k,t}$ will keep their non symmetrical roles.} 
We propose the following model \new{for a given lead time and a given location: for any $e,k,t$,}

\begin{align}
\label{Eq-GammaNormal}
 \begin{cases}
\left(X_{e,k,t}|Z_t\right) & =a_{e}+b_{e}Z_t+c_{e}\varepsilon_{e,k,t} \, \\
 \new{\left(Y_t|Z_t\right)=\left(X_{0,t}|Z_t\right)} &  \new{=a_{0}+Z_t+\varepsilon_{0,1,t}} \\
\left(\varepsilon_{e,k,t}|\omega_t^{-2}\right) & \overset{ind}{\sim}\mathcal{N}\left(0,\omega_t^{2}\right)  \\
\left(Z_t|\omega_t^{-2}\right) & \overset{ind}{\sim}\mathcal{N}\left(0,\lambda\omega_t^{2}\right)\\
\omega_t^{-2} & \overset{iid}{\sim} \Gamma\left(\alpha,\beta\right)
\end{cases},
\end{align}
where $X_{e,k,t}$ denote the $k^{th}$ member of the ensemble $e$ \new{at time $t$}, 
$Z_t$ and $\omega_t^2$ are the \new{corresponding} latent variables (forming the bedrock of the exchangeability property) upon which the ensemble members of a given ensemble $e$ are conditionally independent. These latent variables $Z_t$ and $\omega_t^2$ are assumed to be \nnew{independent across time.} 
$\Gamma\left(\alpha,\beta\right)$ is the gamma distribution with parameters
$\alpha$ and $\beta$ where $\alpha$ and $\beta$, as well as $\lambda$ and 
$\left\{a_e,b_e,c_e\right\}_{e\in\left\{0,\ldots,E\right\}}$, are parameters to be estimated.  \new{These parameters are then specific to the considered lead-time and location}.  
Identifiability constraints impose $b_0=c_0=1$.
Figure~\ref{Fig-DAG} shows the corresponding directed acyclic graph (DAG) for the model.
\begin{figure}[!h]
\centering
 \includegraphics[viewport=450bp 25bp 820bp 345bp,clip,scale=0.3]{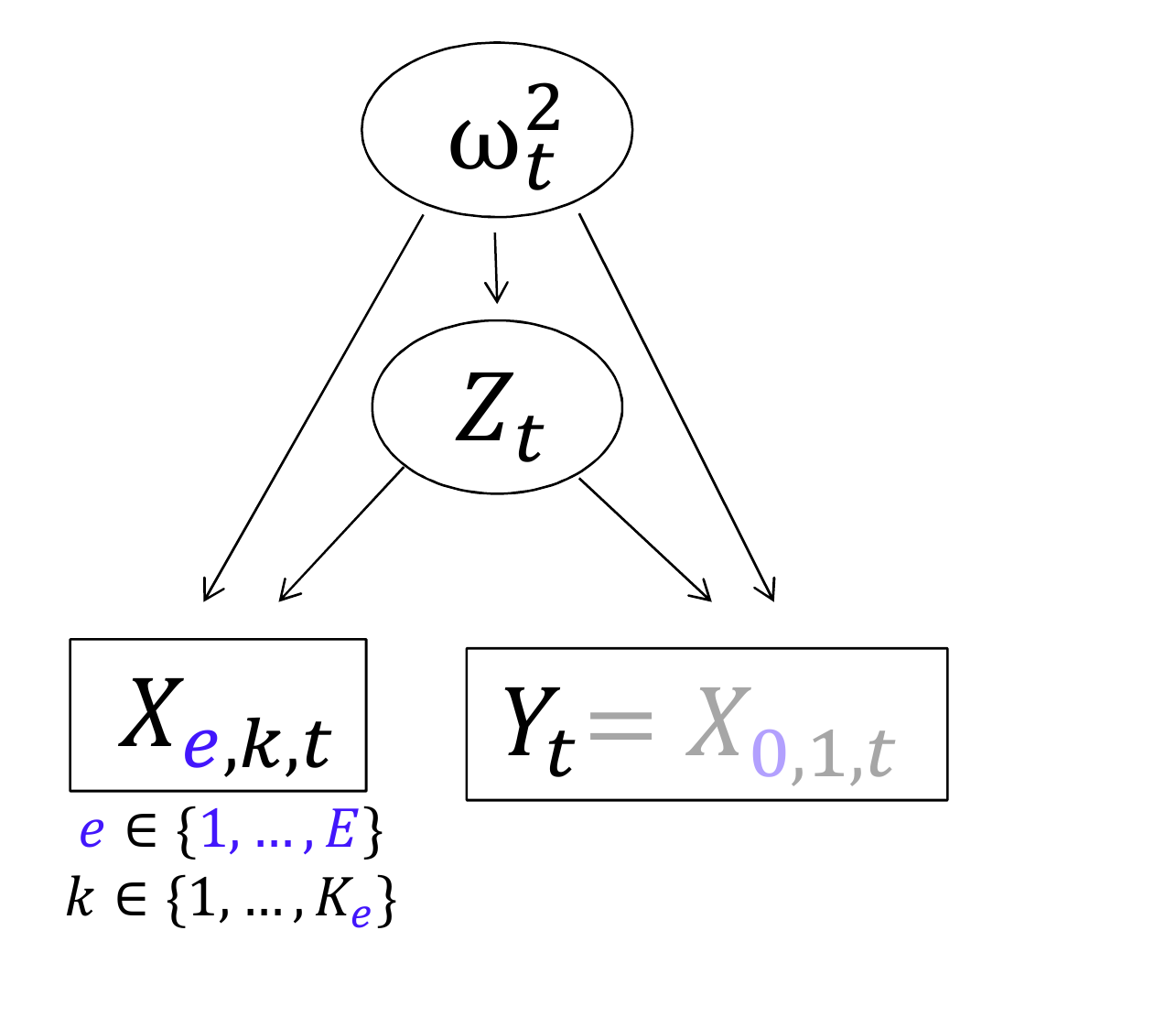}
 \caption{Directed acyclic graph of the model given by {Eq~\ref{Eq-GammaNormal}}.
 $X_{e,k,t}$ is the $k^{th}$ member of the ensemble $e$ \new{at time $t$} and $K_e$ 
 is the number of members from the ensemble $e$, $Y_t=X_{0,K_0=1,t}=X_{0,t}$ relates to the targeted variable to forecast, and $Z_t$ and $\omega_t^2$ are the model backbone latent variables. \new{Times are assumed to be independent.}}
 \label{Fig-DAG}
\end{figure}

The first latent variable, $Z_t$, can be interpreted as some \textit{hidden global state of the atmosphere}, as \textit{seen} 
by a meteorological simulator issuing a forecasting ensemble.
We make the additional assumption that this pivotal
quantity is the same for all ensembles $e=1,\ldots,E$. It makes sense to think that dispersed members yield a large uncertainty on the latent variable $Z_t$. 
\new{Since this dispersion is not constant over time, the second latent variable, $\omega_t^2$, is useful in the model to account for this variation.
}
\new{This is related to using of the variance term in the EMOS model \citep{gneiting2005calibrated}.}
A meteorological interpretation of this second latent variable
$\omega_t^2$ would be something like the underlaying  \textit{turbulent atmospheric condition} (as \textit{encoded} by all meteorological simulators). 
The \textit{adhoc} dependence between $Z_t$ and $\omega_t^2$ as specified by Eq~\ref{Eq-GammaNormal} greatly facilitates inference 
(through the use of Gamma-Normal conjugacy) \new{and therefore leads to a fast algorithm, which is useful in an operational context, where inference can be conducted within a moving window}.
The meaning of parameters $a,b,c,\alpha,\beta $ from the model given by Eq \ref{Eq-GammaNormal} is straightforward: 
\begin{itemize}
\item The difference  $a_{e}-a_{0},\, e>0 $ gives the additive bias for the forecasting ensemble $e$, to be compared to $0$.
\item The ratio $\frac{b_e}{b_0},\, e>0$ is the multiplicative bias of the forecasting ensemble $e$. Since $b_{0}=1$ for identifiability, the value ${b_e}$ is directly to be compared to $1$.
Additive and multiplicative biases may partly compensate one another.
\item For parameter $c$,  the ratio $\frac{c_e}{c_0}=c_{e},\, e>0$ ( parameter $c_0$ being fixed to $1$)  will be understood as a dispersion bias for the predictors.
A ratio greater than $1$ can be interpreted as an over-dispersion of the predicting ensemble $e$ .
\item The ratio $\frac{\beta}{\alpha-1}$ corresponds to the expected value of $\omega_t^2$ which rules how far the quantity to forecast $Y_t$ can occur from the latent variable $Z_t$. 
It is therefore expected that this ratio will increase with the lead time of the forecast, because ensembles generally become less and less informative when the forecasting horizon grows.\\
\end{itemize}

The model given in this section fulfills parsimony and integration of multiple sources of information: each additional ensemble only needs three parameters to be included within the multi-ensemble gamma Normal exchangeable model given by 
Eq~\eqref{Eq-GammaNormal}.

\subsection{Inference}

In what follows, the parameters to be estimated are denoted 
by $\btheta=(\alpha,\beta,\lambda,\ba,\bb,\bc)$ where $\ba=(a_e)_{e\in\{0,\ldots,E\}}$, $\bb=(b_e)_{e\in\{1,\ldots,E\}}$ and $\bc=(c_e)_{e\in\{1,\ldots,E\}}$ , recalling that $b_0=c_0=1$. 

We moreover use Gelfand's bracket notations for probability distributions \citep{gelfand1990sampling} and  we denote by 
$(\bbX,\bY)$ 
the set of predictors $\bX_t$ and observations $Y_t$ acquired over time during the learning period 
\new{and $\bZ$ and $\bO^{2}$ the sets of latent variables}.

Assuming that the parameters remain the same over a learning period close to or homogeneous to the prediction period, the EM algorithm \citep{dempster1977maximum} is an effective instrument for estimating the parameters of this multivariate normal model with random effects.
The E-step is tractable since, for any $t$, the conditional distribution 
$\left[Z_t,\omega_t^{-2}|\bX_t,Y_t;\btheta\right]$
follows a normal-gamma distribution. 
In the M-step, some parameter updatings have explicit formulas and another relies on a numerical optimization procedure. 
The proof and explicit formulas for updating the parameters are provided in Appendix \ref{sec:appendice:inference}.
We denote by 
$\btheta^{(h)}=\left(\alpha^{(h)},\beta^{(h)},\lambda^{(h)},\ba^{(h)},\bb^{(h)},\bc^{(h)}\right)$ the current value of the parameters.

\begin{algorithm}[H]
  \SetSideCommentLeft
  \DontPrintSemicolon
  \KwSty{Initialization:} Set $\btheta^{(0)}$ by a method of moments.\;

  \Repeat{$\left\|\btheta^{(h+1)} - \btheta^{(h)}\right\| < \varepsilon$}{
   
   \begin{enumerate}
   \item[E-step] Compute needed moments of latent variables in $\bZ$ and $\bO^2$ with respect to the current \\
   parameter values ($\btheta^{(h)}$).
   \item[M-step] Update the current parameter values:
   $$\btheta^{(h+1)}=
   \arg\max_{\btheta}\E_{[\bZ,\bO^{-2}|\bbX,\bY;\btheta^{(h)}]}\left(\log\left(\left[\bbX,\bY,\bZ,\bO^{-2};\btheta\right]\right)\right)$$
   \end{enumerate}
   
  }
 \caption{EM algorithm for estimating parameters in Model \eqref{Eq-GammaNormal}}
 \label{algo:EM:gaussian}
\end{algorithm}

\subsection{Forecasting}
For a new time $t^{\prime}$ with predictors $\mathbf{X}_{t'}$, the forecast for $Y_{t'}$ is provided by 
the predictive distribution of
$(Y_{t^\prime}|\mathbf{X}_{t'})$ \new{(the forecasting target in the operational systems considered here),} which is given in the next proposition.  

\begin{proposition}
\label{prop:forecastnormal}
Under Model \ref{Eq-GammaNormal}, for a new time $t'$, the predictive distribution $\left(Y_{t^\prime}|\mathbf{X}_{t'}=\mathbf{x}_{t'}\right)$ follows a Student distribution with scale parameter
$\sqrt{\frac{(\lambda^{\prime\prime}+1)\beta_{t^\prime}^{\prime\prime}}{\alpha^{\prime\prime}}}$, location parameter $a_{0}+m_{t^\prime}^{\prime\prime}$ and $2\alpha^{\prime\prime}$ degrees of freedom where
\begin{equation*}
\begin{array}{lcl}
\alpha^{\prime\prime} & =&\alpha+\frac{\sum_{e=1}^{E}K_{e}}{2},\\
\lambda^{\prime\prime-1} & =&\sum_{e=1}^{E}K_{e}b_{e}^{2}c_{e}^{-2}+\lambda^{-1},\\
m_{t^\prime}^{\prime\prime} & =&\lambda^{\prime\prime}\cdot\sum_{e=1}^{E}c_{e}^{-2}b_{e}K_{e}\left(\bar{x}_{e,t^\prime}-a_{e}\right),\\
\beta_{t^\prime}^{\prime\prime} & =&\beta+\frac{1}{2}
\left\{\sum_{e=1}^{E}\sum_{k=1}^{K_{e}}c_{e}^{-2}(x_{e,k,t^\prime}-a_{e})^{2}-
m_{t^\prime}^{\prime\prime2}\lambda^{\prime\prime-1}\right\}.
\end{array}
\end{equation*}
\end{proposition}

\begin{proof}
By using the conjugacy properties of the normal gamma model as in the E-step of Algorithm \ref{algo:EM:gaussian}, we obtain that 
 $\left(Z_{t'}|\omega^{-2}_{t'},\mathbf{X}_{t'}\right)$ follows a normal distribution 
$\mathcal{N}(m^{\prime\prime}_{t^\prime},\lambda^{\prime\prime}\omega_{t'}^{2})$ and
 $\left(\omega_{t'}^{-2}|\mathbf{X}_{t^\prime}\right)$ follows a gamma distribution
$\Gamma\left(\alpha^{\prime\prime},\beta_{t^\prime}^{\prime\prime}\right)$ where parameters are determined by identification.
Moreover, we have from Model \eqref{Eq-GammaNormal}:
\begin{eqnarray*}
 \left(Y_{t'}|Z_{t'}\right)&=&a_0+Z_{t'}+\varepsilon_{0,1,t'}\\
 \left(\varepsilon_{0,1,t'}|\omega_{t'}^2\right)&\sim&\mathcal{N}\left(0,\omega_{t'}^2\right)\,.
\end{eqnarray*}
Then, $\left(\frac{Y_{t'}}{\sqrt{\lambda^{\prime\prime}+1}}|\omega^2_{t'},\mathbf{X}_{t'}\right)\sim
\mathcal{N}\left(\frac{a_0+m_{t^\prime}^{\prime\prime}}{\sqrt{\lambda^{\prime\prime}+1}},\omega_{t'}^2\right)$
and by using the distribution of  $\left(\omega_{t'}^{-2}|\mathbf{X}_{t'}\right)$ we obtain the announced result.
\end{proof}

\begin{remark}
The predictive mean and the predictive variance are given respectively by:
\begin{equation}
\begin{array}{lcl}
\mathbb{E}\left(Y_{t'}|\mathbf{X}_{t'}=\mathbf{x}_{t'}\right)&=&
a_{0}+\lambda^{\prime\prime}\sum_{e=1}^{E}\frac{b_{e}}{c_{e}^{2}}K_{e}\left(\bar{x}_{e,t^\prime}-a_{e}\right),\\
 \mathbb{V}\left(Y_{t'}|\mathbf{X}_{t'}=\mathbf{x}_{t'}\right)
 &\propto&\beta+\frac{1}{2}\left\{\sum_{e=1}^{E}c_{e}^{-2}\sum_{k=1}^{K_{e}}(x_{e,k,t^\prime}-a_{e})^{2}-
 \lambda^{\prime\prime}\left\{\sum_{e=1}^{E}\frac{b_{e}}{c_{e}^{2}} K_{e}\left(\bar{x}_{e,t^\prime}-a_{e}\right)\right\}^{2}\right\}\,.
\end{array}
\label{eq-esperance}
\end{equation}
\end{remark}
This predictive distribution is very similar to the EMOS one, (see Eq~\eqref{eq:EMOS}), for which the predictive expectation is expressed linearly as a function of the mean of the members and the predictive variance as a function of the ensemble variance.
In Eq~\eqref{eq-esperance}, it appears that a member from ensemble $e$ has an even greater impact on the forecast as the $\frac{b_e}{c_e^2}$ ratio is large. 
Therefore, we define the \textit{contribution} of a member of the ensemble $e$ to the final forecast by: 
\begin{equation*}
contrib_e=\frac{\frac{b_e}{c_e^2}}{\sum_{e^\prime=1}^{e^\prime=E}K_{e^\prime}\frac{b_{e^\prime}}{c_{e^\prime}^2}}.
\end{equation*}


\section{An extension of the multi-ensemble exchangeable Gamma Normal model to the precipitation case}
\label{sec:Tobit-case}
In this section, we investigate an adaptation of the post-processing method based on the exchangeability hypothesis for precipitation-like variables. These variables cannot be assumed to be normally distributed:
they exhibit a mixed nature with a discrete component at zero and a positive continuous component.
In the long term, we wish to be able to jointly post-process temperature variables and rainfall type variables,
this is one of the reasons why we seek to remain within the convenient framework of the Gaussian family.
The approach proposed herein can be viewed as a Tobit regression (Tobin, 1958 ; Chib, 1992) and is similar to techniques presented in \cite{scheuerer2015statistical} and \cite{thorarinsdottir2010probabilistic}.

\subsection{Multi-ensemble Multilevel Exchangeable Tobit model}
\label{Subsec:ModelAllard}
The underlying idea of the following model is that precipitation (as observations with a zero discrete component and a continuous positive component),
would be some left censorship from a continuous (latent) variable \citep{allard2012modeling}.
Based on this idea, some work has been already undertaken to develop a post-processing method for precipitation forecasts by \citep{schultz2008continuous}.
In this work, the latent continuous variable associated with precipitation has a physical meteorological interpretation and is called 
\textit{pseudo-precipitation}. We leave aside any physical interpretation and assume that these pseudo-precipitations are Gaussian after an appropriate invertible transformation $f_\mathcal{N} $ such as the Box-Cox one \citep{box1964analysis}.
The model we propose for such normalized pseudo-precipitation is the same as that proposed for temperature variables.

Let $Y^\prime_t $ and $\bX_t^\prime $ be the precipitation to forecast and its predictors.
 $Y_t$ and $\bX_t$ become in this section latent variables: they correspond to the underlaying normal pseudo-precipitation and its normal pseudo-ensembles of predictors.
The model is as follows \new{for a given lead time and a given location: for any $e,k,t$}: 
\begin{align}
\begin{cases}
\left(X_{e,k,t}|Z_t\right) & =a_{e}+b_{e}Z_t+c_{e}\varepsilon_{e,k,t}\\
\new{\left(Y_t|Z_t\right)=\left(X_{0,1,t}|Z_t\right)} & \new{=a_{0}+Z_t+\varepsilon_{0,1,t}}\\
\left(X_{e,k,t}^{\prime}|X_{e,k,t}\right) & =\mathbb{I}_{X_{e,k,t}>\nu}f_\mathcal{N}^{-1}\left(X_{e,k,t}\right)\\
\left(\varepsilon_{e,k,t}|\omega_t^{-2}\right) & \underset{}{\overset{ind}{\sim}}\mathcal{N}(0,\omega_t^{2})\\
\left(Z_t|\omega_t^{-2}\right) & \overset{ind}{\sim} \mathcal{N}(0,\lambda\omega_t^{2})\\
\omega_t^{-2} & \overset{iid}{\sim} \Gamma\left(\alpha,\beta\right)
\end{cases}
\label{eq:tobit}
\end{align}
$\mathbb{I}$ denotes the indicator function.
 $Z_t$ and $\omega^2_t$ are the latent backbone of this multilevel exchangeable model and are assumed to be \nnew{independent across time.}
Conditionally upon ($Z_t,\omega_t^2$), the ensemble members are iid within each ensemble. 
$\alpha$, $\beta$, $\lambda$, $\nu$ and 
$\left\{a_e,b_e,c_e\right\}_{e\in\left\{0,\ldots,E\right\}}$ 
are parameters to be estimated \new{(specifically to each lead time and location)}.
Again for identifiability concerns,  $b_0=c_0=1$. 
The random variables $\left(X_{e,k}\right)_{e\in0,\ldots ,E\:k\in1,\ldots ,K_{e}}$ are now latent (yet observed when greater than $\nu$).
Figure~\ref{Fig-DAGprecip} shows the DAG (Directed Acyclic Graph) that corresponds to the precipitation model.
\begin{figure}[!h]
\centering
 \includegraphics[viewport=50bp 170bp 430bp 570bp,clip,scale=0.3]{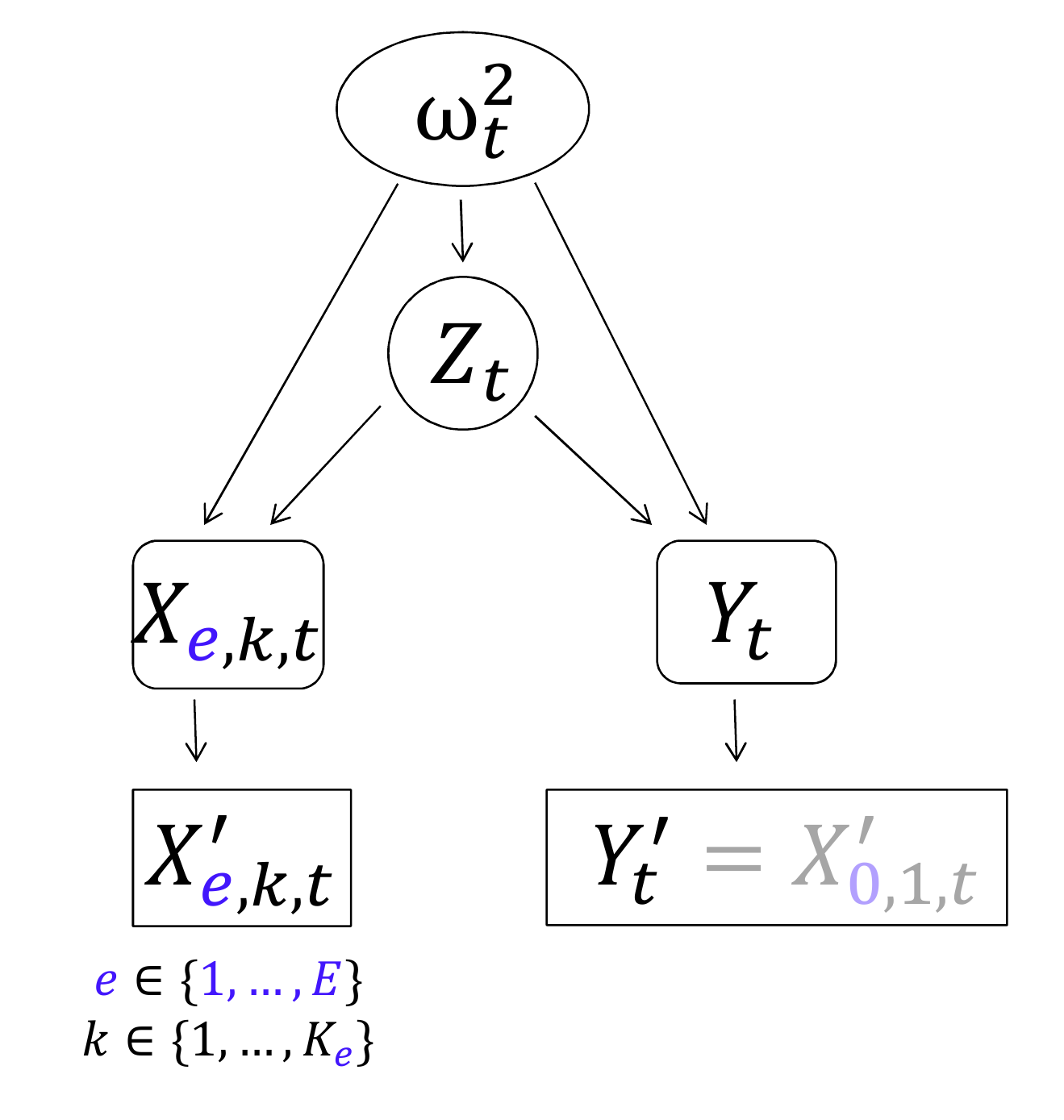}
 \caption{Direct acyclic graph of the post-processing model for precipitation. 
 $X^{\prime}_{e,k,t}$ denotes the $k^{th}$
 member of ensemble $e$ \new{at time $t$},
 $K_e$ is the size of the ensemble $e$, 
 $Y^{\prime}_t=X^{\prime}_{0,K_0=1,t}=X^{\prime}_{0,t}$ denotes the variable to forecast 
 and $\mathbf{X_t}$, $Y_t$, $Z_t$ and $\omega_t^2$ are latent variables
 (partially observed in the case of $\mathbf{X_t}$ and $Y_t$). \new{Times are assumed to be independent.}}
 \label{Fig-DAGprecip}
\end{figure}

\subsection{Inference}
As for the model of the previous section, the EM algorithm is tailored for estimating the parameters of our exchangeable Tobit model. Beforehand, we proceed with the estimation of the normalizing transformation from historical precipitation data.

\subsubsection*{Normalisation parameters}
The selection of an appropriate normalizing transformation $f_\mathcal{N}$ is an important issue. In this work, we consider the power transform: $f_\mathcal{N}(x^\prime)=x^{\prime\gamma}$, where $\gamma$ is a parameter to be estimated. We choose the same transformation parameter regardless of the precipitation ensemble considered, the one estimated from the observed precipitation values (variable to predict). Thus, in our model, $f_\mathcal{N}$ does not depend on the ensemble $e$. Consequently, we can use all historical precipitation data available for inferring parameter $\gamma$.
Obviously, for precipitations, it is essential to ensure that the inverse transformation, $f_\mathcal{N}^{-1}$, takes positive values on  $\left]\nu,+\infty\right[$. In this work, we set 
$\nu=0$ which
corresponds to a very refined sensitivity of the rain gauge, with non-zero values of observed rainfall very close to $0$. 
The goal is therefore to find a value for $\gamma$ such that 
 $Y^{\prime\gamma}$
can be assumed to follow a normal distribution and is left censored at $0$.

Assuming such a model with temporal independence of the precipitation phenomenon, the log-likelihood, can be written as follows:
\begin{align*}
 \mathcal{L}\left(\{y'_t\}_t\right)&=\sum_{t\textrm{ s. t. }y{\prime}_t>0}
 \left(\log\left[Y_t^{\prime}=y_t^{\prime}\right]\right)+\#\{t, y'_t=0\} \log([Y^{\prime}_t=0]) \\
 &=\sum_{t\textrm{ s. t. }y{\prime}_t>0}\left( \log\left[Y_t^{\prime\gamma}=y_t^{\prime\gamma}\right]
 +\left(\gamma-1\right)\log\left(y^\prime_t\right)+\log\left(\gamma\right)\right)+\#\{t, y'_t=0\} 
 \log(\left[Y_t\le 0\right])\\
 &=\sum_{t\textrm{ s. t. }y{\prime}_t>0} \left(\log\psi\left(y^{\prime\gamma}_t;\mu,\sigma^2\right)
 +\left(\gamma-1\right)\log\left(y^\prime_t\right)+\log\left(\gamma\right)\right)+\#\{t, y'_t=0\}\log (\Psi(0;\mu,\sigma^2)),
\end{align*}
where $\mu$ et $\sigma$ are also parameters to be estimated,
$\psi\left(x;\mu,\sigma^2\right)$ and $\Psi\left(x;\mu,\sigma^2\right)$ are respectively the pdf and the cumulative distribution function at $x$ of a Gaussian distribution with mean $\mu$ and variance $\sigma ^2$
and $\#$ denotes the cardinal of a set.
In practice, parameters $\mu, \sigma ^2$ and $\gamma$ are obtained by maximizing the likelihood with a numerical optimization method (the Nelder-Mead procedure implemented in \texttt{R}). The power transformation 
is then applied to the observations and the corresponding ensemble forecasts.

\subsubsection*{Other parameters: the Stochastic EM algorithm}
The E-step of the EM algorithm requires to compute, for any $t$, the conditional distribution function of 
$(\bX_t,Y_t,Z_t,\omega_t^{-2}|\bX'_t,Y_t')$,
which is not  tractable.
The distribution of $\left(Z_t,\omega_t^{-2}|\bX_t,Y_t,\bX_t^{\prime},Y_t^{\prime}\right)$  is the same as the distribution of
$\left(Z_t,\omega_t^{-2}|\bX_t,Y_t\right)$
and the 
distribution of
$\left(\bX_t,Y_t|Z_t,\omega_t^{-2},\bX_t^{\prime},\
Y_t^{\prime}\right)$ is given by: for any $e,k,t$ \new{(included $e=0$ that is $X_{0,1,t}=Y_t$)},
\begin{align}
\left[X_{e,k,t}|Z_t,\omega_t^{-2},X_{e,k,t}^{\prime}
\right] =& \begin{cases}
 & \mathbb{I}_{\left\{X_{e,k,t}=f_\mathcal{N}\left(X_{e,k,t}^{\prime}\right)\right\}}\:\textrm{if \ensuremath{X_{e,k,t}^{\prime}>0}}\\
 & \psi_{<\nu}\left(X_{e,k,t};a_{e}+b_{e}Z_t,c_{e}^{2}\omega_t^{2}\right)\:\textrm{if \ensuremath{X_{e,k,t}^{\prime}=0}}
\end{cases}
\label{eq:loiXcondX}
\end{align}
where $\psi_{<\nu}\left(x;\mu,\sigma^{2}\right)$ denotes the Gaussian pdf with mean $\mu$ and variance $\sigma^2$ truncated to the right at $\nu$. 
Therefore, we can add a simulation step (S-step) before the E- step in the inference algorithm.
This leads to a partially stochastic EM algorithm \citep{broniatowski1983reconnaissance,celeux1985sem}.
The SEM algorithm and the Gibbs algorithm used in S-step are provided hereafter.

\begin{algorithm}[H]
  \SetSideCommentLeft
  \DontPrintSemicolon
  \KwSty{Initialization:} Set $\btheta^{(0)}$ by a method of moments.\;
  \Repeat{$\left\|\btheta^{(h+1)} - \btheta^{(h)}\right\| < \varepsilon$ }{
   
   \begin{enumerate}
   \item[S-step] 
   For each time $t$ of the learning set, Simulate $(\bX_t^S,Y_t^S)$ with respect to the conditional\\ distribution 
$\left[\bX_t,Y_t|\bX_t^{\prime},Y_t^{\prime};\btheta^{(h)}\right]$ by Algorithm \ref{algo:gibbs}.
   \item[E-step] Compute needed moments of latent variables in $\bZ$ and $\bO^2$ with respect to the current \\ parameter values ($\btheta^{(h)}$) and $(\bbX^S,\bY^S)$.
  \item[M-step] Update the current parameter values:
   $$\btheta^{(h+1)}=
   \arg\max_{\btheta}\E_{[\bZ,\bO^{-2}|\bbX^S,\bY^S;\btheta^{(h)}]}\left(\log\left(\left[\bbX^S,\bY^S,\bZ,\bO^{-2};\btheta\right]\right)\right)$$
   \end{enumerate}
   
  }
 \caption{SEM algorithm for estimating parameters in Model \eqref{eq:tobit}}
 \label{algo:SEM:preci}
\end{algorithm}

\begin{algorithm}[H]
  \SetSideCommentLeft
  \DontPrintSemicolon
  For a current value of the parameters $\btheta^*=
  \left(\alpha^{*},\beta^{*},\lambda^{*},\ba^{*},\bb^{*},\bc^{*}\right)$.\;
  \KwSty{Initialization:} Simulate $\omega^{-2,(1)}$ from $\Gamma(\alpha^{*},\beta^{*})$
then $Z^{(1)}$ from $\mathcal{N}(0,\lambda^{*}\omega^{2,(1)})$.\;
  \For{$i$ in 1 to $N_{iter}$}{
   
   \begin{enumerate}
   \item For each $(e,k)$ such that $X'_{e,k}=0$, simulate $X^{(i)}_{e,k}$ from 
   $\left[X_{e,k}|Z^{(i)},\omega^{-2,(i)},X_{e,k}^{\prime};\btheta^*\right]$.
   \item Simulate $(Z^{(i+1)},\omega^{-2,(i+1)})$ 
   from $\left[Z,\omega^{-2}|\bX^{(i)},Y^{(i)};\btheta^*\right]$.
   \end{enumerate}\;
 \KwSty{Return:} $(\bX^S,Y^S)=(\bX^{(N_{iter})},Y^{(N_{iter})})$.
  }
 \caption{Gibbs algorithm for simulating latent variables $(\bX,Y)$ conditionally to $(\bX',Y')$}
 \label{algo:gibbs}
\end{algorithm}

\bigskip
Note that in Algorithm $\ref{algo:gibbs}$, for each $t$, the values of $X_{e,k,t}$ corresponding to $X'_{e,k,t}\not=0$ are set as $X_{e,k,t}=f_\mathcal{N}(X'_{e,k,t})$. The number of iterations $N_{iter}$ has to be chosen large enough to ensure that the Gibbs algorithm has converged. 

\subsection{Forecasting}
In order to simulate $Y'_{t'}$ from the  conditional distribution of $\left(Y'_{t'}|\mathbf{X}'_{t'}\right)$, a  sample from the distribution of $\left(Z_{t'},\omega^{-2}_{t'}|\mathbf{X}_{t'}'\right)$ is first drawn. Then, for each couple $\left(Z_{t'},\omega^{-2}_{t'}\right)^{(i)}$ of this sample, 
we simulate $Y_{t'}^{(i)}$ from the distribution of $(Y_{t'}|Z_{t'}^{(i)},\omega^{-2,(i)}_{t'})$ and apply $Y^{\prime(i)}_{t'}=f_\mathcal{N}^{-1}\left(Y_{t'}^{(i)}\right)\mathbb{I}_{Y_{t'}^{(i)}>\nu}$. The sample of the distribution of $\left(Z_{t'},\omega^{-2}_{t'}|\mathbf{X}_{t'}'\right)$ is simulated by using Algorithm \ref{algo:gibbs}, where $Y'_{t'}$ is not considered in the conditional distributions in Steps 1 and 2. In this sample, we removed the first iterations of the Gibbs Algorithms which correspond to a burn-in period.

\section{Simulation study }
\label{sec:Numerical-experiment}
In this section we check that, in a realistic framework, i. e. for parameters close to those that would be learned from real data sets, we are able to correctly estimate them with the SEM algorithm described above.
Since the $\gamma$ parameter of the normalizing transformation has been learned separately,  we work directly in the normalized space.
For this algorithmic experiment, we simulate artificial data according to the Tobit model described in Section \ref{Subsec:ModelAllard} with the following parameters:
\begin{align*}
  K & = \left(1,10,35,1\right)\\
  a & = \left(a_0,a_1,a_2,a_3\right) = \left(0,1,0.7,-0.1\right)\\
  b & = \left(b_0,b_1,b_2,b_3\right) = \left(1,1.1,1,0.9\right)\\
  c & = \left(c_0,c_1,c_2,c_3\right) = \left(1,0.8,0.7,1.1\right)\\
  \alpha & = 2.5, \quad \beta=3, \quad \lambda=0.5,
\end{align*}
where the first element of each vector is relative to the observations (referred to with index $0$).
We generate $100$ artificial datasets of $200$ elements each. 
The SEM algorithm described above is then run on the $100$ first elements of each of the $100$ datasets for inference with the initial value of the parameters chosen by the method of moments.
Note that parameters $b_0$ and $c_0$ are not estimated: they are set to $1$.
The SEM algorithm is launched for $1000$ steps and, within each iteration, the Gibbs algorithm carries out $4$ iterations, which appears to be sufficient in this case given the good estimation results. 

The distributions of parameter estimates obtained from the $100$ simulated datasets are illustrated in the form of boxplots on Figure~\ref{Fig-EstimParamCens}.
Whatever the parameter considered, the resulting estimate does not show any significant bias.  
\begin{figure}[h!]
\centering
 \includegraphics[scale=0.7]{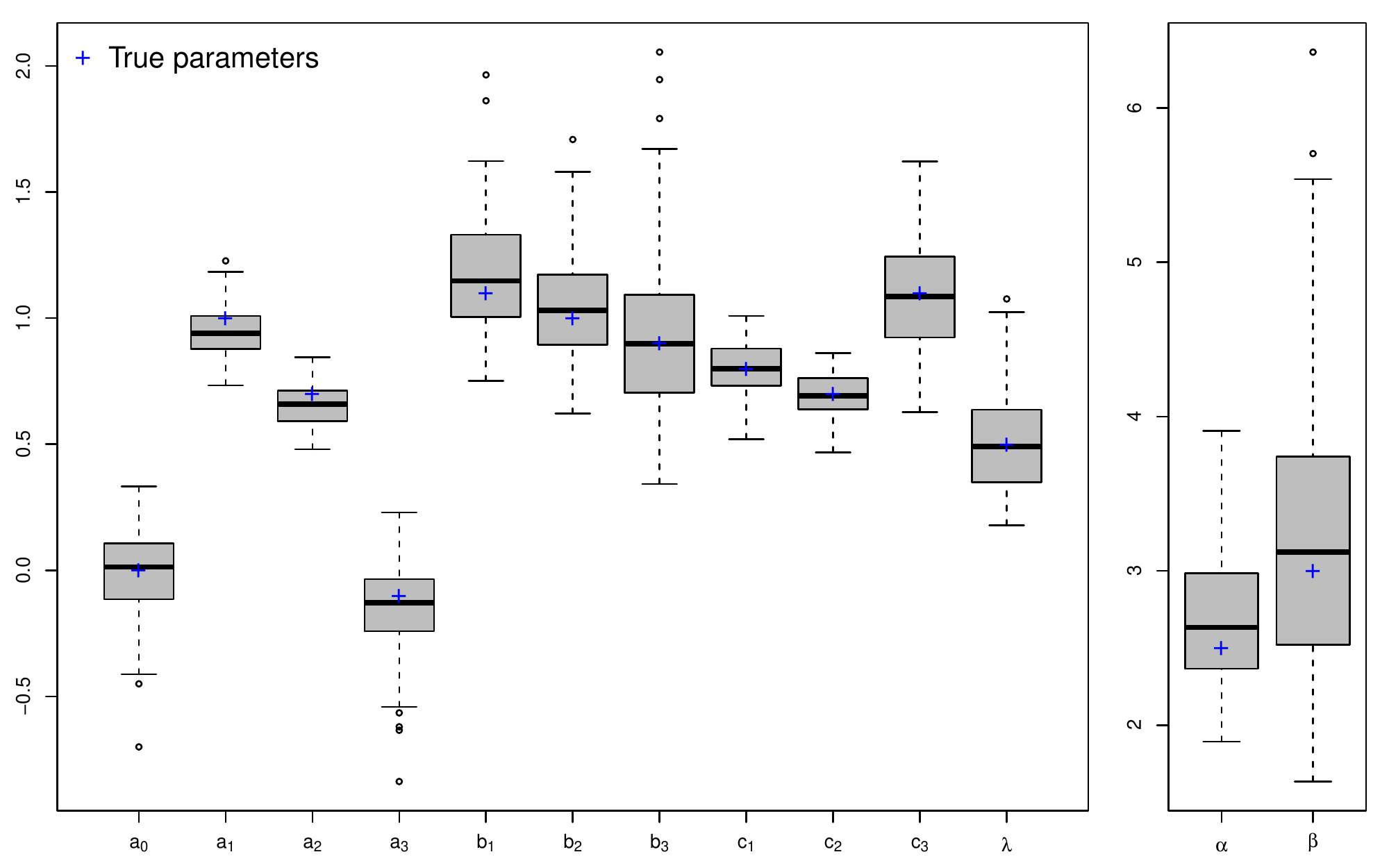}
 \caption{
 Parameters obtained after inference from $100$ simulations of $100$ repetitions according to the precipitation model. The parameters used for the simulation are represented by the blue crosses.}
 \label{Fig-EstimParamCens}
\end{figure}

\nnew{We then assess the performance of the corresponding forecasts. 
We run our forecasting method on the $100$ last elements of 
each of the $100$ simulated data sets, both with the parameters used for simulations ("oracle" forecasts) and with the estimated parameters ("prediction" forecasts). We compare those forecasts to those obtained by considering the raw ensembles as samples from probabilistic forecasts.
Guided by \citep{gneiting2007probabilistic}, we apply verification tools such as the rank histogram and the continuous ranked probability score (CRPS) to evaluate the performance of those forecasts. The rank histogram assesses their  reliability, while the CRPS evaluates both their reliability and their sharpness.
A rank histogram is computed from the forecasts on each of the $100$ datasets. All these ranks histograms are summarized by the p-values obtained by comparing them to flat histograms through multinomial goodness-of-fit Chi-squared tests. The results of those
tests show a good reliability of our forecasts and an even better reliability of the oracle forecasts, as illustrated in Figure~\ref{fig:simuforecast} (left-hand-side). 
The reliability of Ensemble~3 (which is a single member ensemble) is omitted. 
The CRPS (Figure~\ref{fig:simuforecast}, right-hand-side) shows a better performance of our forecasts compared with the ones obtained from the raw ensembles. Note that in the case of the third ensemble (deterministic forecast), the CRPS reduces to the Mean Absolute Error (MAE) \citep{hersbach2000decomposition}.
We also check the reliability of the simulated distributions of the latent variables $Z$ and $\omega^2$ by computing their coverage rates. The median coverage rates of the 88\% credible intervals are respectively of $0.70$ and $0.87$ for prediction and oracle methods in the case of $Z$, and of $0.78$ and $0.86$ in the case of $\omega^2$.
}

\begin{figure}[h!]
    \centering
\begin{minipage}{.47\textwidth}
 \includegraphics[scale=0.56]{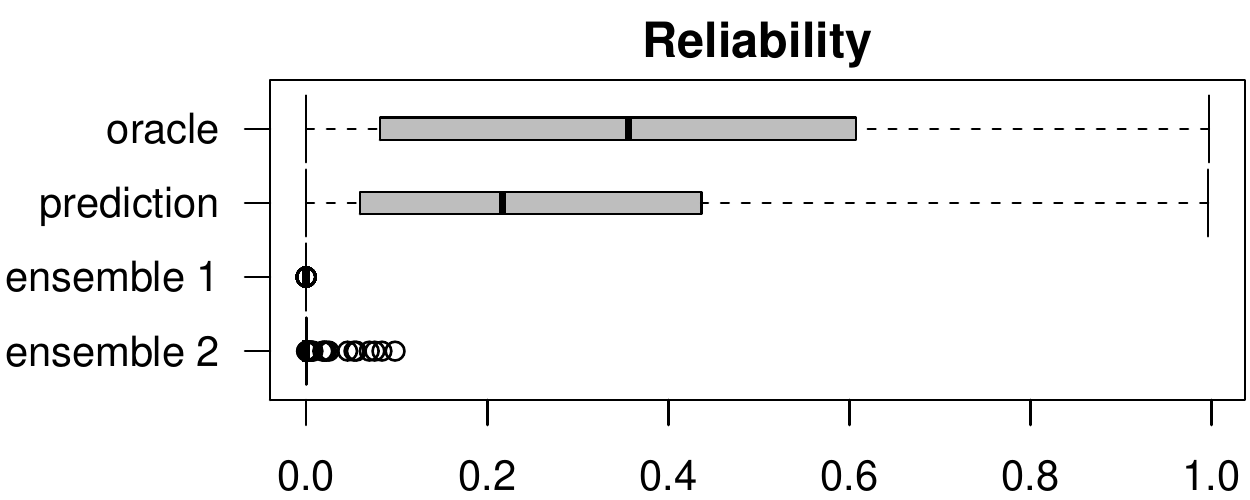}
\end{minipage}
\begin{minipage}{.47\textwidth}
\includegraphics[scale=0.56]{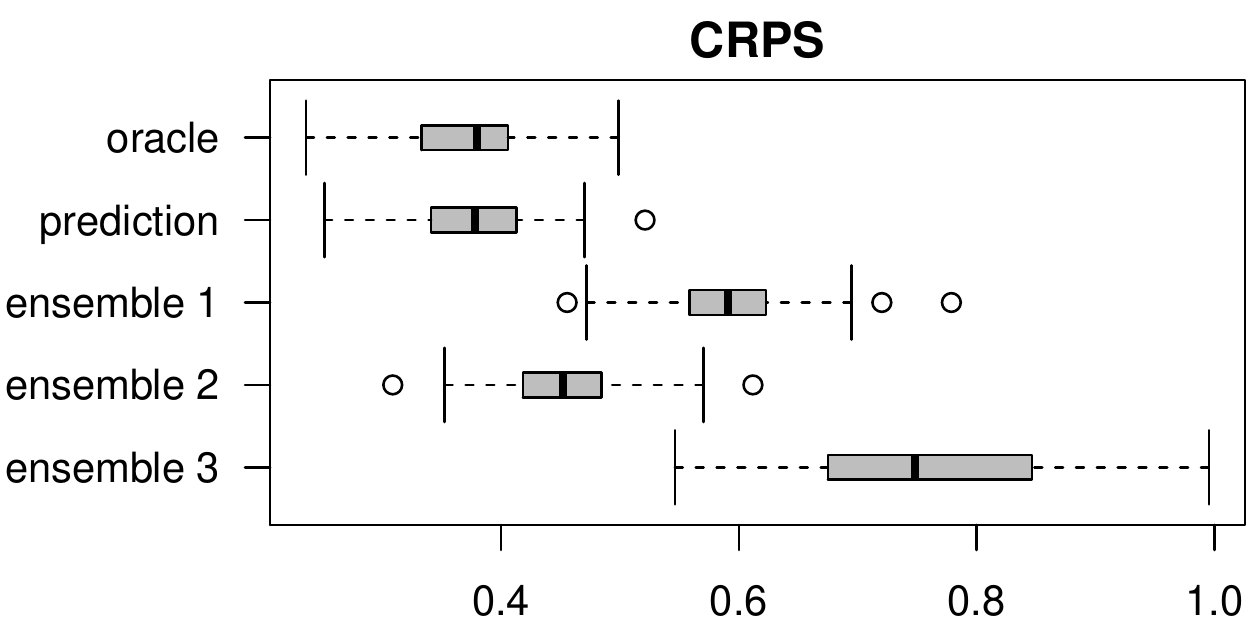}
\end{minipage}
    \caption{\nnew{Reliability (left-hand-side) and CRPS (right-hand-side) of forecasts issued by post-processing approaches computed over $100$ datasets simulated according to the precipitation model. Reliability is here summarized by the p-values of the multinomial goodness-of-fit Chi-squared tests which compare the rank histograms to a flat histogram. \textit{oracle} and \textit{prediction} refer to our post-processing method with respectively true and estimated parameters. They are compared with the three raw ensembles considered as probabilistic forecasts (Reliability of Ensemble $3$ is omitted since it is a single member ensemble).
    }}
    \label{fig:simuforecast}
\end{figure}

\section{Application to meteorological data from Québec}
\label{sec:Applications}

The case study developed in this section to illustrate the use of our post-processing approaches concerns ensemble temperature and precipitation forecasts available daily over Hydro-Québec Manicouagan watershed, a major hydropower system. As illustrated on Fig~\ref{Fig-ComplexeManic}, the Manicouagan watershed is subdivided into five subcatchments for which meteorological forecasts are used every day to produce streamflow predictions:  upstream to downstream, Manic-5, Petit Lac Manicouagan, Toulnoustouc, Manic-3 and Manic-2. The Manicouagan watershed is located in northeast of the province of Québec, Canada. This water resources system consists of two hydropower plants with reservoirs in parallel (Manic-5 and Toulnustouc) and three downstream run-of-river hydropower plants (Manic-3, Manic-2 and Manic-1). 
\begin{figure}[h!]
\centering
 \includegraphics[scale=0.4]{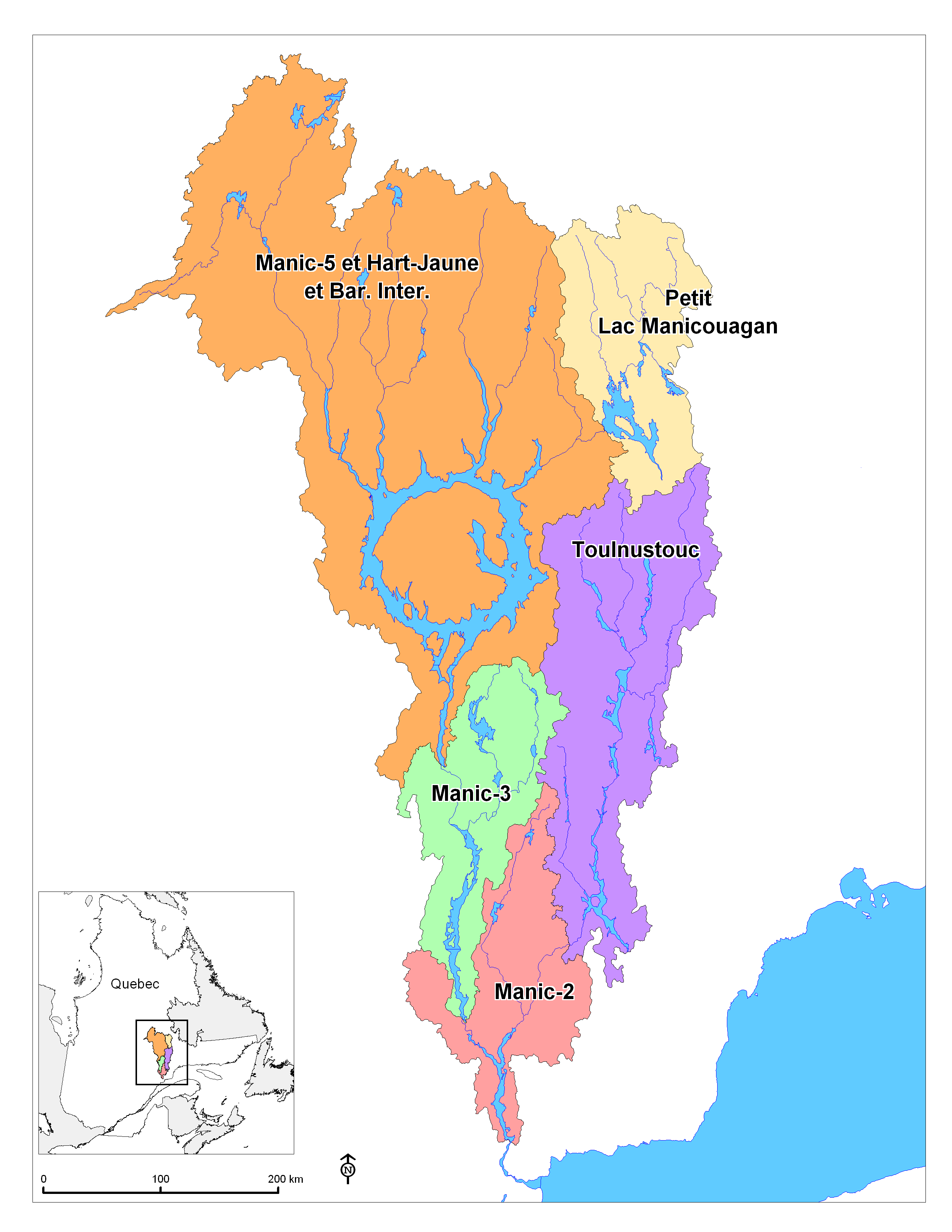}
 \caption{Manicouagan watersheds in Québec, Canada.}
 \label{Fig-ComplexeManic}
\end{figure}
The total installed capacity is 6,202 MW, which is about 17 pourcent of Hydro-Québec's total capacity.  In operating the system, generation planners face a variety of decisional problems. Two of those, common to every installation, are safety and the respect of environmental laws and regulations.  For the two upstream watersheds, with large reservoirs, the other main concerns are those of long term energy planning and optimization, and efficient releases for the operation of the run-of-river plants, given the inflows on those sub-basins. For the three run-of-river plants, the issue is an efficient scheduling, given the inflows on the watersheds and the upstream releases. It is quite clear that a good prediction of the future state of inflows, which highly depends upon weather forecasts, plays a major role on the decisions that will be made, and the efficiency of the operations. The need for reliable ensemble weather forecasts is thus self-evident.

To assess the performance of the post-processing mixed effect models, Hydro-Qu\'ebec provided us with records of daily meteorological forecasts and corresponding observations for the five watersheds for years 2013 and 2014. The meteorological ensemble forecasts were extracted from the THORPEX Interactive Grand Global Ensemble (TIGGE) database (\citep{park2008tigge}). Three daily ensembles for forecasting lead times ranging from $1$ to $9$ days produced by meteorological global forecast centres were considered:

\begin{itemize}
\item the CMC-EPS (Ensemble Prediction System), from the Canadian Meteorological Center (CMC), with $20$ ensemble members,
\item the NCEP-GEFS (Global Ensemble Forecasting System) from the National Centers for Environmental Prediction (NCEP), with $20$ ensemble members,
\item the ECMWF-EPS from the European Center for Medium-Range Weather Forecasts \\ \mbox{(ECMWF)}, with $50$ ensemble members.
\end{itemize}
Meteorological variables of interest for Hydro-Québec are
daily minimal and maximal temperatures, and precipitations. Since the rainfall-runoff model in use is a lumped and conceptual hydrological model \citep{guay2018hsami1}, the raw ensemble forecasts available at grid points belonging to the basins under study, or located adjacent to it, have been averaged to get global watershed values. These are to be compared to the 
corresponding observed values computed by Hydro-Québec.

In this section, we focus on maximal temperatures 
to test our Gamma Normal model, 
and on precipitations, to test our Tobit model. Both models assume that ensemble members are exchangeable 
within each ensemble.
This assumption seems appropriate for the ECMWF-EPS and the 
NCEP-GEFS, given the way their ensemble members are produced, 
but it is not for the CMC-EPS.
We thus look for sub-ensembles of members of 
the CMC-EPS within which exchangeability can be assumed. 
A rank test 
shows that, in the case of precipitations,  
even ensemble members and odd ensemble members 
constitute two appropriate subgroups \new{(see section \ref{sec:Exchangeability})}, \new{we thus treat them as two separate ensembles}. In the case of maximal temperatures, however, we could not find 
any such sub-group. In the absence of a better solution, 
we consider the whole ensemble.

Figure~\ref{Fig-ExManic2-raw} shows an example of a forecasting situation that has been treated in this case study, a forecast produced on the 30th of April 2014. Raw CMC-EPS, NCEP-GEFS and ECMWF-EPS precipitations and maximal temperature ensemble 
 forecasts for Manic 2 watershed are presented for lead times from 1 up to 9 days. 
The target values, to be predicted, 
are indicated by the dotted black lines.
On the precipitation example (bottom panel), we decompose 
the CMC-EPS into its two exchangeable sub-ensembles, 
odd ensemble members (CMC-EPS-1) and 
even ensemble members (CMC-EPS-2). 

\begin{figure}[!h]
\centering
 \includegraphics[scale=0.5]{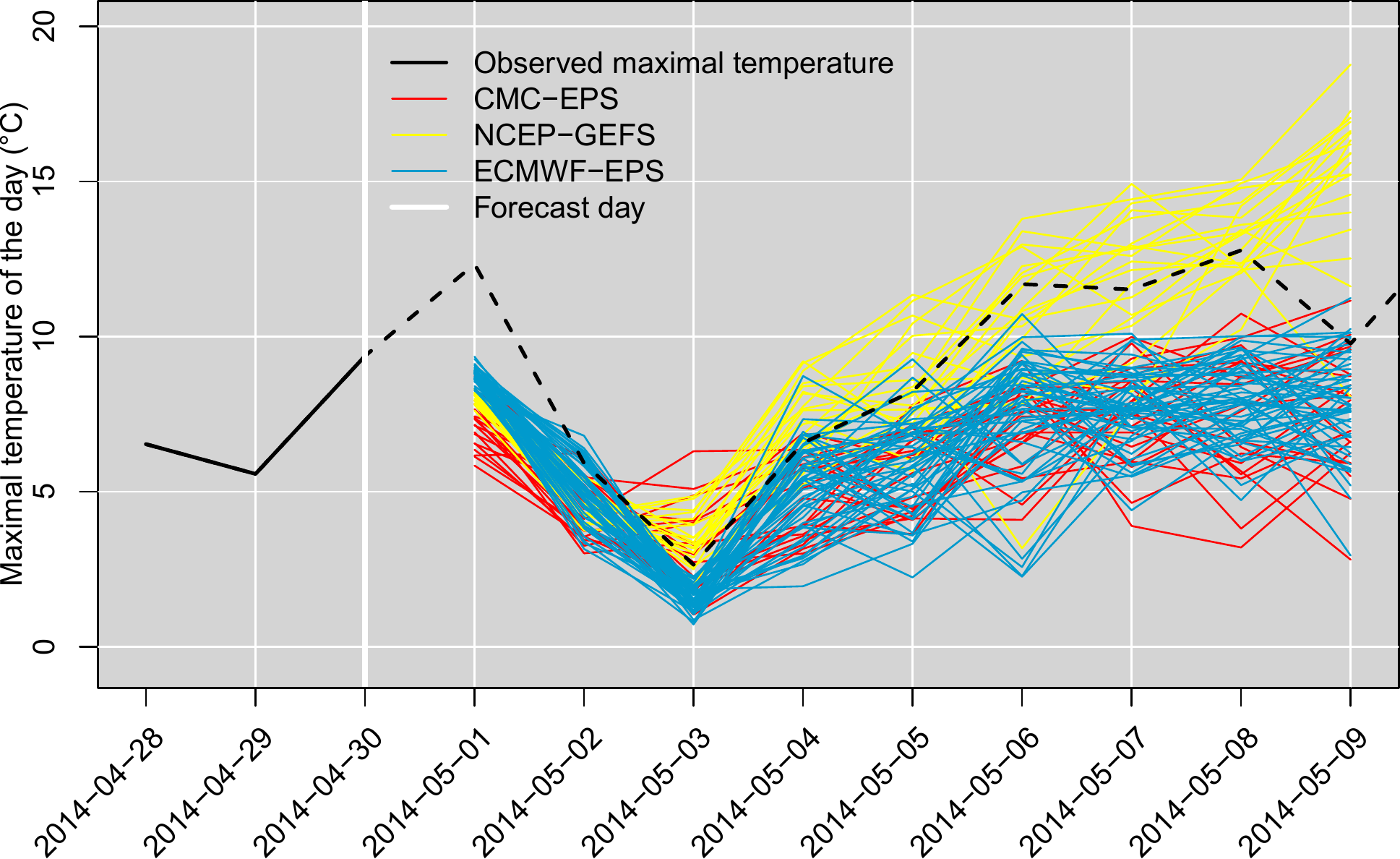}\\
  \includegraphics[scale=0.5]{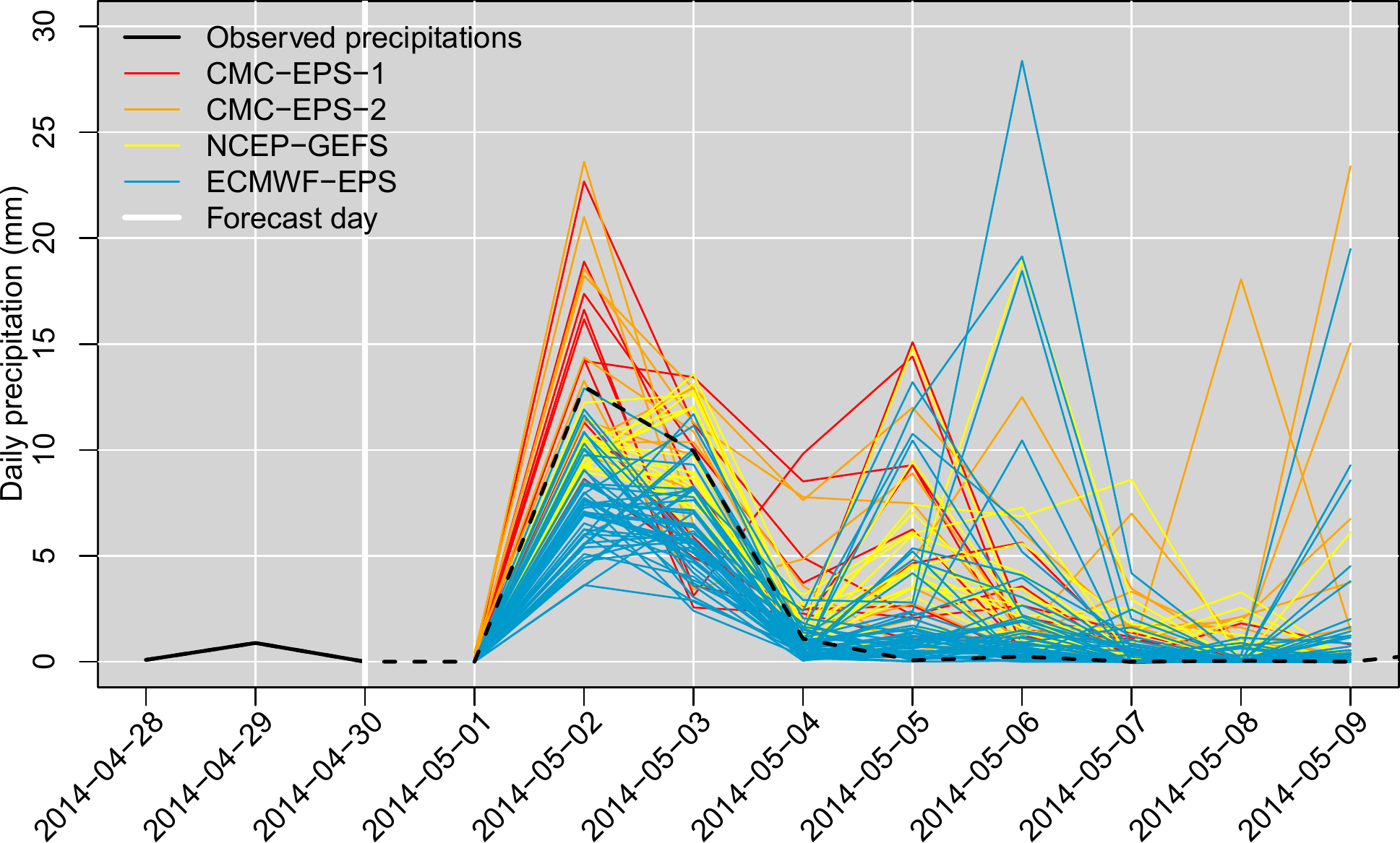}\\
 \caption{Precipitations and maximal temperature raw ensemble forecasts 
on the Manic~2 watershed produced the 30 of April 2014 for the nine 
upcoming days (CMC-EPS, odd members: 1 and even members: 2, NCEP-GEFS and ECMWF-EPS).
The target values, to be predicted, are indicated by the dotted black line.}
\label{Fig-ExManic2-raw}
\end{figure}

In the following applications, we are seeking to produce reliable predictive distributions for each lead time and each meteorological variables. First, the Gamma Normal post-processing model is applied to daily maximal temperatures ensemble forecasts, and then
the Tobit model is applied to precipitations ensemble forecasts.
 As in Section \ref{sec:Numerical-experiment}, we apply verification tools such as the rank histogram and the CRPS to evaluate the performance of the post-processing approaches.
 The CRPS have been calculated from daily predictive distributions produced for 2014, with parameters estimated using observations and raw predictions available for 2013. In both cases, temperatures and precipitations, we first focus on Manic 2 watershed and then extend our results to the four other basins.

\subsection{Application to maximal daily temperatures forecasts}

\subsubsection*{Illustration on the Manic~2 watershed}
\label{Sec-Manic2Tmax}
Figure~\ref{Fig-EstimParamManic2} shows the estimated parameters as a function of the forecast lead times. The graph Fig~\ref{Fig-EstimParamManic2}$\left(1\right)$ illustrates the additive bias of each three ensemble forecasts, given by the difference between parameters $a_1,a_2$ et $a_3$, related to the $3$ ensembles, and  parameter $a_0$, corresponding to observations.
It indicates that the three ensemble prediction systems would have produced negatively biased forecasts during 2013 for the Manic~2 watershed, regardless of the forecasting lead time.
The graph Fig~\ref{Fig-EstimParamManic2}$\left(2\right)$ gives inference results for parameters $b$ that can be interpreted as the forecast multiplicative bias. Its value is less than $1$ in the case of CMC-EPS et ECMWF-EPS ensembles, which further amplifies the diagnosis stemming from the first graph. On the other hand, this figure shows a positive multiplicative bias of the NCEP-GEFS ensemble for lead time higher than six days. This could potentially compensate for the negative additive bias previously observed for 2013.

\begin{figure}[!h]
\centering
 \includegraphics[width=0.45\textwidth]{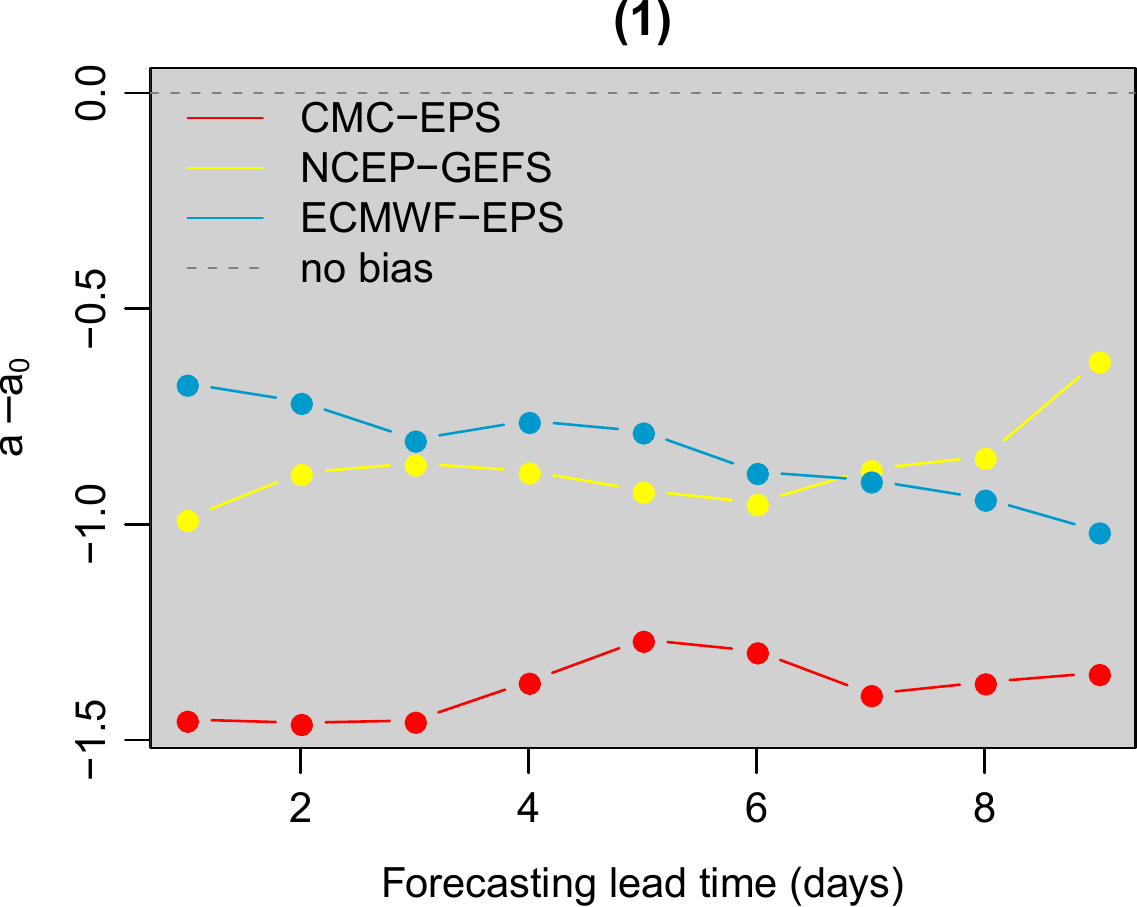}
 \includegraphics[width=0.45\textwidth]{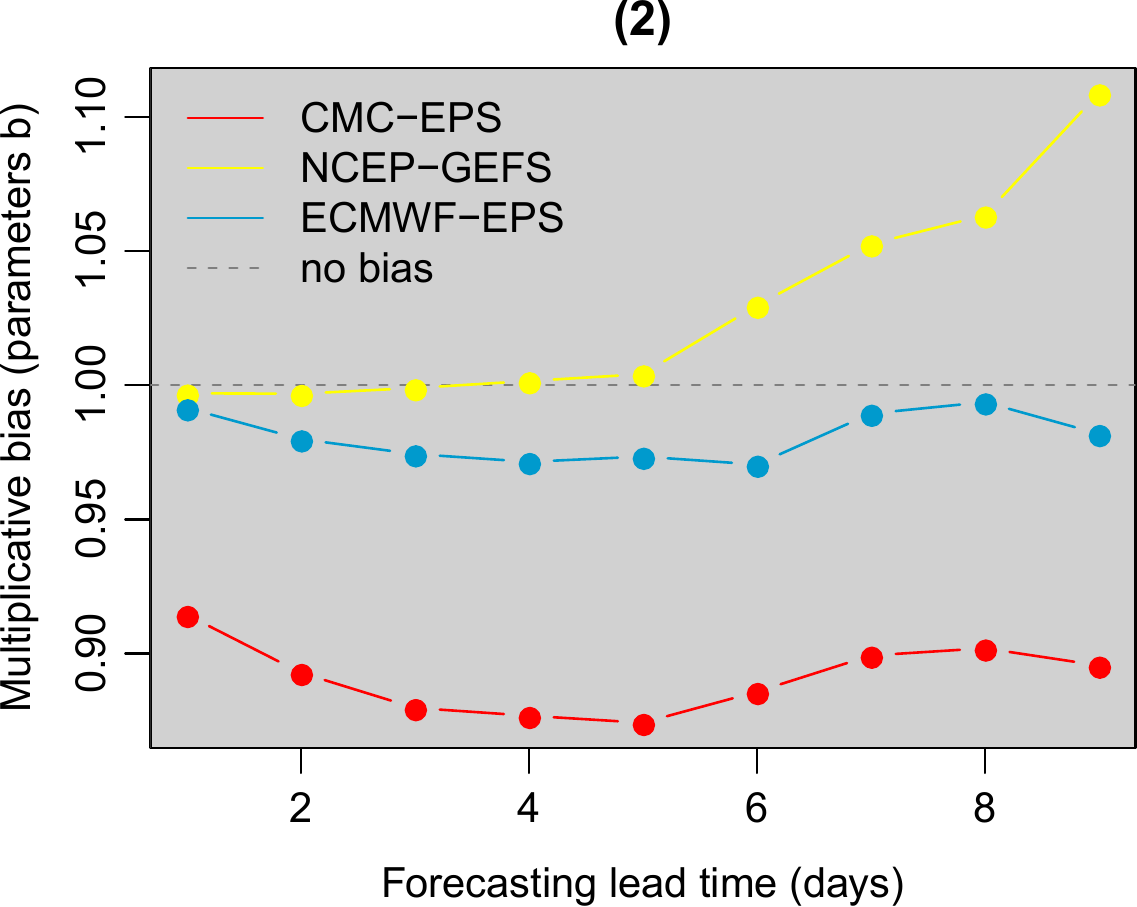}\\
 \includegraphics[width=0.45\textwidth]{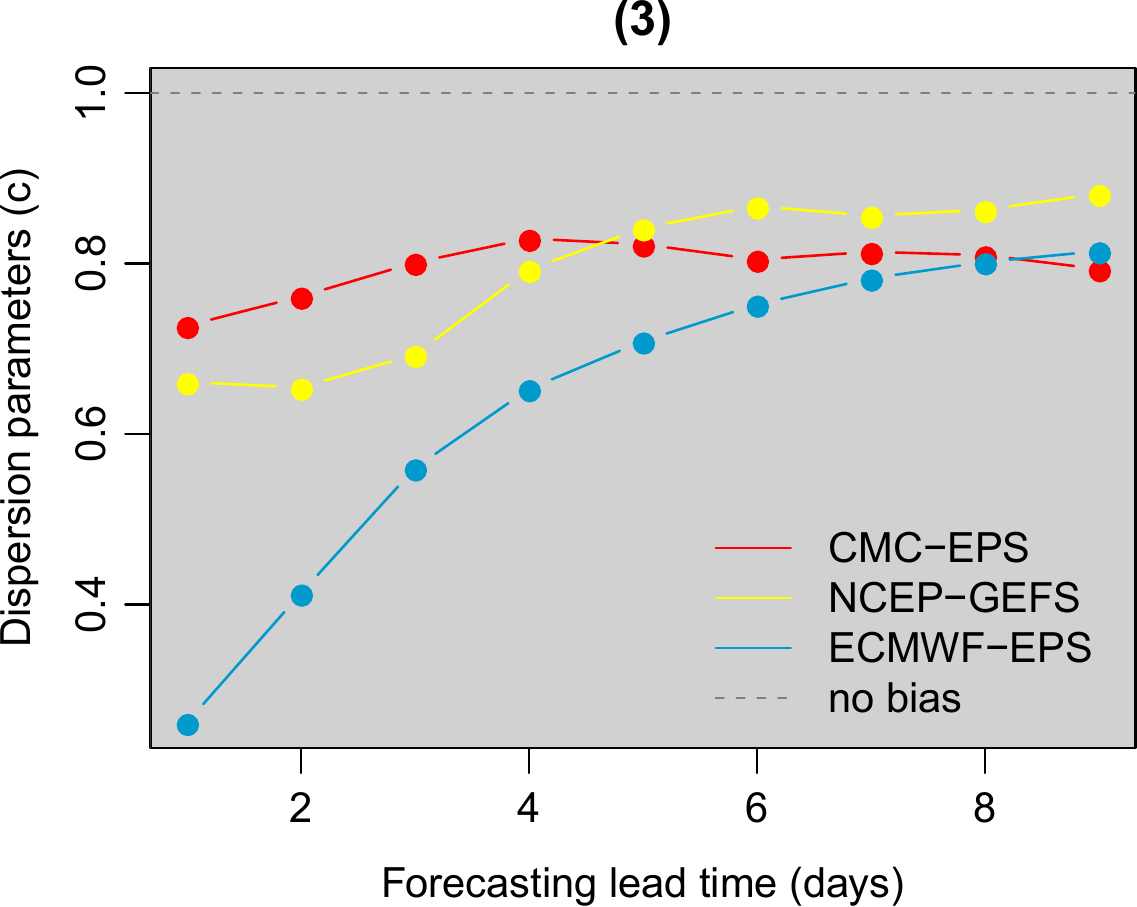}
  \includegraphics[width=0.45\textwidth]{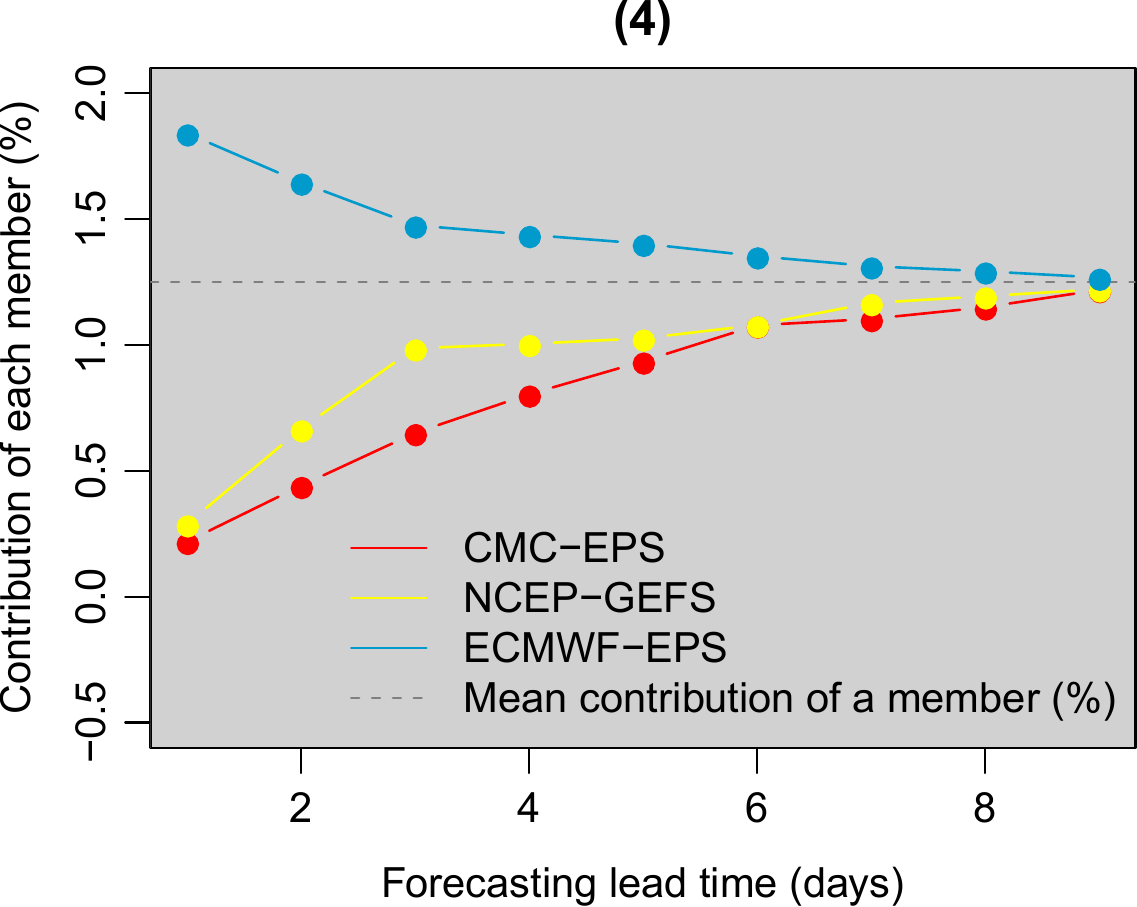}\\
 \includegraphics[width=0.45\textwidth]{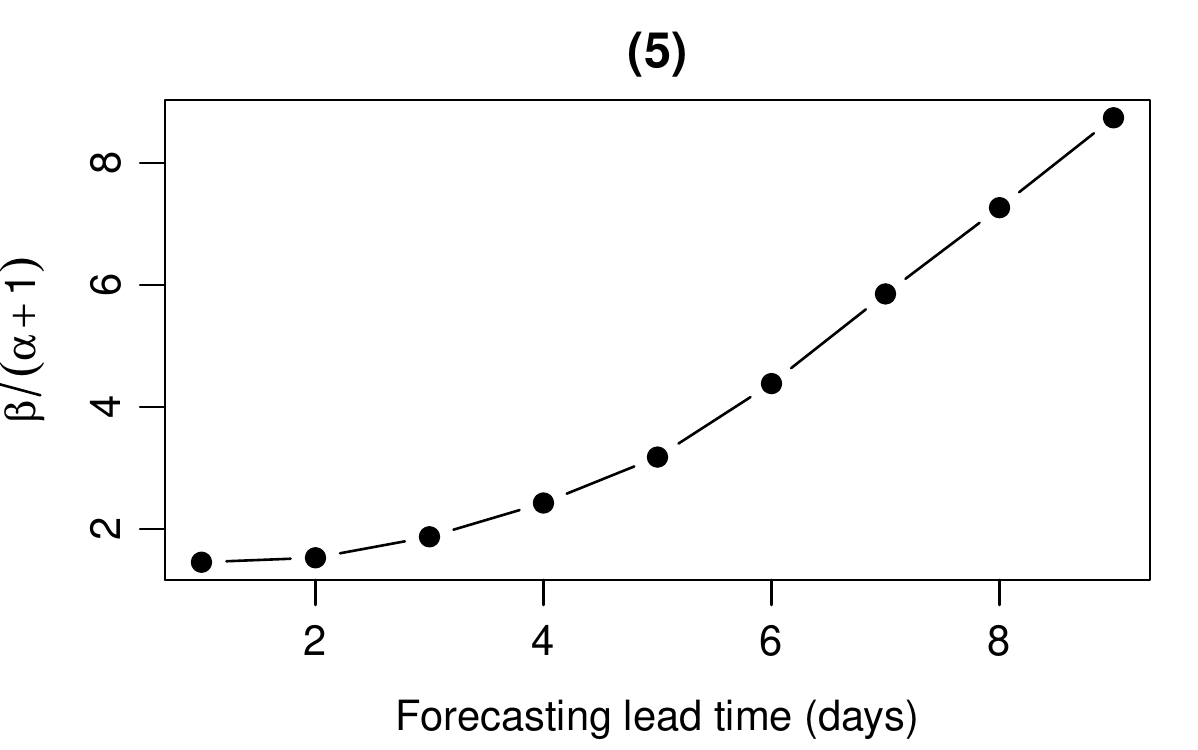}
 \includegraphics[width=0.45\textwidth]{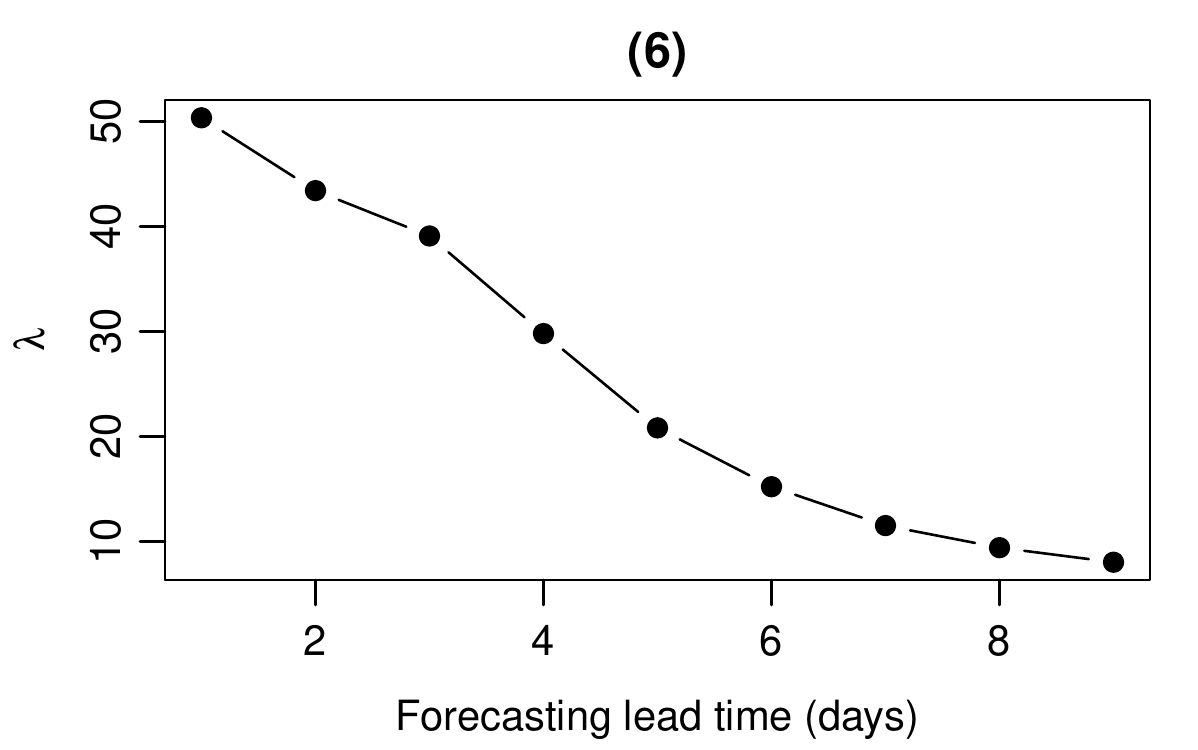}
 \caption{ 
 Parameters obtained with the EM algorithm are represented as a function of the forecasting horizon for the Manic~2 watershed with $E=3$ forecasting sources (CMC-EPS, NCEP-GEFS and ECMWF-EPS). Year 2013 has been used as the learning period.}
 \label{Fig-EstimParamManic2}
\end{figure}

Fig~\ref{Fig-EstimParamManic2}$\left(3\right)$ shows the inference results for parameters $c$ which set the inter-member dispersion of the ensemble relative to the observation one. Note that all values are lower than $1$ and increase with the forecast lead time. Therefore, the ensemble produced in 2013 would have been underdispersed, especially for short-term forecasting.
As well as the $c$ parameters, the ratio $\frac{\beta}{\alpha-1}$, illustrated on Fig~\ref{Fig-EstimParamManic2}$\left(5\right)$, increases with the forecast lead time. This ratio corresponds to the expected value of $\omega_t^2$ and therefore settles the variability of the quantity to be forecasted $Y_t$ around the latent variable $Z_t$. Thus, as expected, the uncertainty blurring $Y_t$ increases with the forecast lead time since forecasts become less informative for longer term prediction.
The graph of Fig~\ref{Fig-EstimParamManic2}$\left(4\right)$ shows the contribution to the final forecast, $contrib_e$, of each member of the ensemble prediction system $e$. It is expressed in percentage of the total contribution to the overall forecast. We observe that for 2013 at Manic 2 watershed the contribution from the members of the ECMWF-PES is much larger than that from the members of the other ensembles. However, forecasts tend to be similar in terms of contributions when the lead time increases.
Recall that the ECMWF ensemble includes $50$ members, while the other ones each include $20$ members. Therefore most of the information comes from the ECMWF ensemble. This observation is also confirmed by the first bar of Figure~\ref{Fig-Poids} where are shown the relative contribution of each forecast source $e$, given by $ K_e\times contrib_e$, averaged over the $9$ forecast lead times.
\\

We now evaluate our post-processing approach for daily maximal temperature ensemble forecasts produced in 2014 at the Manic 2 watershed.
Figure~\ref{Fig-ExManic2} again shows the forecasting situation of Figure~\ref{Fig-ExManic2-raw}, the maximal temperatures forecasts that have been produced on April 30th, 2014.

\begin{figure}[!h]
\centering
 \includegraphics[scale=0.5]{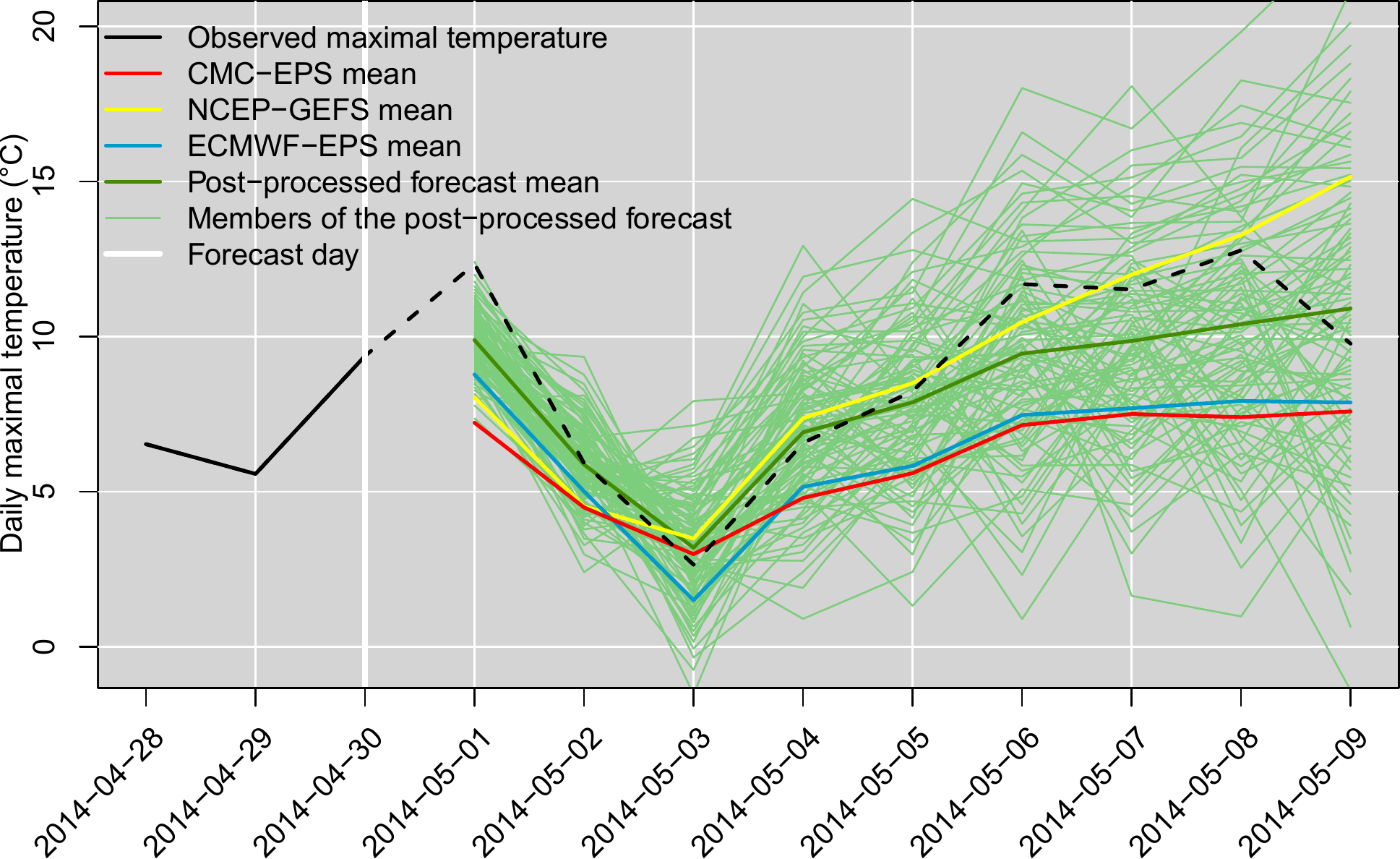}
 \caption{
 Example of a forecast issued for the daily maximum temperatures of the Manic~2 catchment area with the proposed post-processing method taking NCEP-GEFS, CMC-EPS and ECMWF-EPS as inputs. Predictive scenarios derived from meteorological forecasts by forecast time horizon using the ECC-Q are presented.
 The maximum daily temperature to be forecast (observed afterward) is indicated by the black dotted line.}
 \label{Fig-ExManic2}
\end{figure}
The graph shows the averages of the three raw ensembles considered (blue for ECMWF-EPS, yellow for NCEP-GEFS and red for CMC-EPS). The average of the forecasts resulting from post-processing involving these three forecast sources is also presented in green. For this specific forecast situation, the resulting prediction is clearly a compromise between the ECMWF-EPS and CMC-EPS members on the one hand, and NCEP-GEFS on the other.
The 90 scenarios of the forecast illustrated by green fine lines are obtained from the post-processed forecasts initially produced independently, lead time by lead time, these outputs being re-ordered using Ensemble Copula Coupling (ECC).  
\new{This approach consists in  
copying the rank structure observed in a raw ensemble to produce scenarios from independent points. 
For instance, if the scenario predicting the highest temperature at time 1 also predicted the lowest temperature at time 2 in the raw ensemble, these two extremes will also be linked in one of the final predictive scenarios.
In the ECC-Q version of the ECC, used here, the starting points are quantiles of the marginal predictive distributions \citep{schefzik2013uncertainty}.}

Based on rank histograms, the 2014 daily post-processed forecasts obtained for the Manic~2 catchment can be considered reliable, regardless of the forecast lead time (figure not shown). The CRPS values presented in Fig~\ref{Fig-CRPSManic2} support these results. This figure shows the following comparisons~: 

\begin{itemize}
\item the forecasts obtained by considering the raw CMC-EPS ensemble, \texttt{CMC-EPS},
\item the forecasts obtained by post-processing the CMC-EPS with our method after learning the model parameters on year 2013 (in this case, $E=$1), \texttt{Post-processed CMC-EPS},
\item the forecasts obtained by post-processing the CMC-EPS with the standard EMOS method, described in \cite{gneiting2005calibrated}, \texttt{Post-processed CMC-EPS with EMOS}
\item the forecasts obtained by considering the members of the CMC-EPS, NCEP-GEFS and ECMWF-EPS ensembles combined in a large ensemble,
\texttt{Grand Ensemble},
 \item the forecasts obtained by post-processing the CMC-EPS, NCEP-GEFS and ECMWF-EPS ensembles, considered as 3 distinct forecast sources ($E=3$) according to the post-processing method proposed herein whose corresponding parameters are illustrated in Fig~\ref{Fig-EstimParamManic2}, \texttt{Post-processed Grand Ensemble}.
\end{itemize}

It is seen that the raw maximal temperature ensembles of 2014 have higher CRPS values for all lead times compared to the corresponding post-processed ones. The raw CMC-EPS forecasts, currently used by Hydro-Québec for inflow forecasting, are significantly improved (according to CRPS) by the proposed post-processing method.  
One can also observe that its performance is of the same order of magnitude as that of the standard EMOS method, which was expected since these methods give similar forecasts. Furthermore, multi-ensemble post-processing makes it possible to improve the forecasts of the large raw ensemble, in particular for the shorter lead times. Finally, Figure~\ref{Fig-CRPSManic2} shows, for the Manic-2 catchment, that the combination of several sources of forecasts with statistical post-processing would be the option to be favored. This configuration indeed obtains the smallest CRPS values.

\begin{figure}[!h]
\centering
\includegraphics[scale=0.5,page=1]{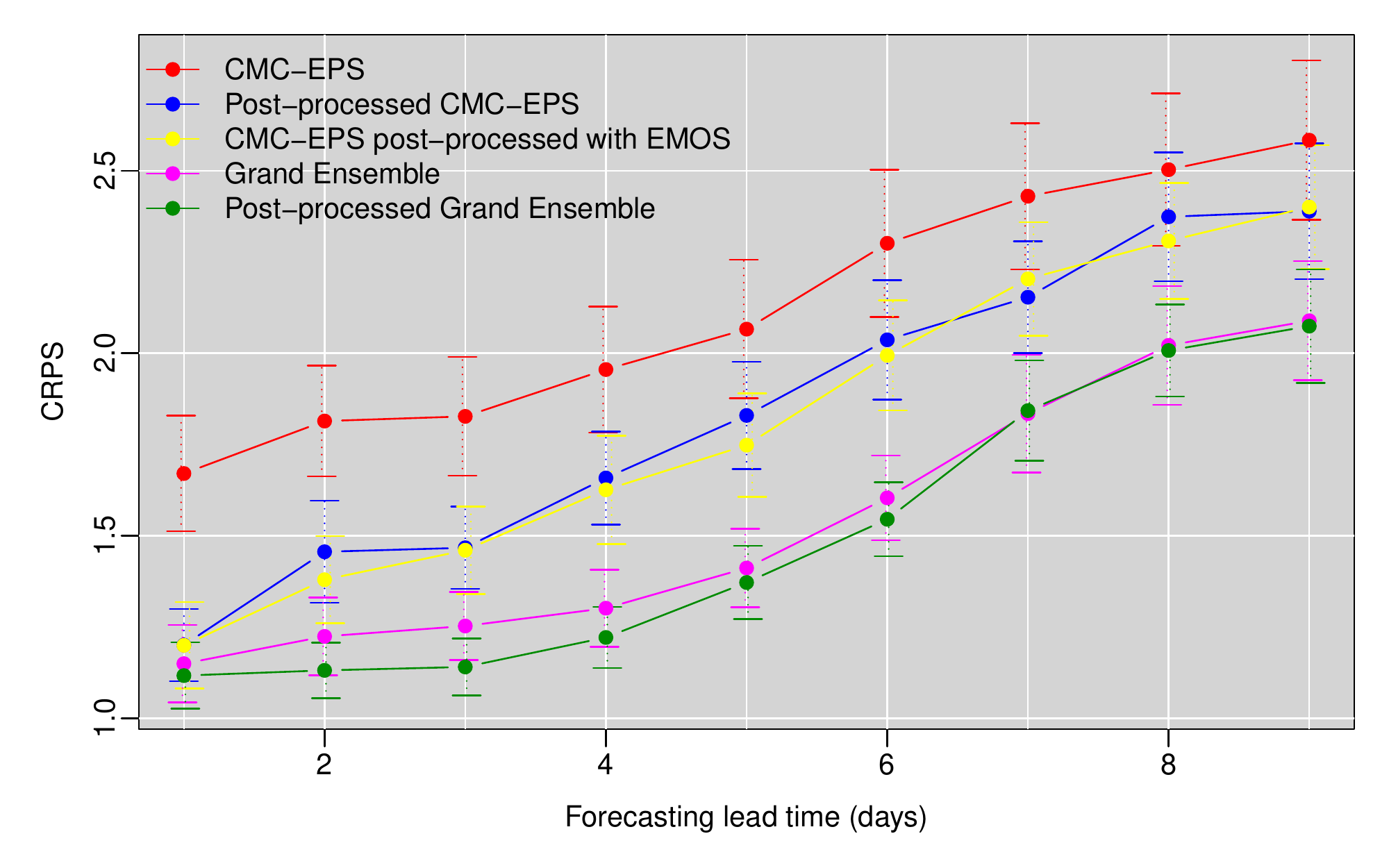}
\caption{CRPS (in $^{\circ}$C) for maximum temperature forecasts at Manic-2 catchment for all 9 lead times with associated bootstrap intervals. The parameters have been estimated using year 2013 and verification has been performed for year 2014.}
\label{Fig-CRPSManic2}
\end{figure}

\subsubsection*{Extending the results to all watersheds}

Based on rank histograms (not shown here), the post-processed maximal temperature daily ensemble forecasts for 2014 can be considered reliable regardless of the catchment and the lead time. According to the CRPS, post-processed ensembles are generally better than raw ensembles. Figure~\ref{Fig-CRPS5bassins} reports the average CRPS values as a function of lead times for each of the five watersheds studied. The main observations drawn from this figure are as follows:

\begin{itemize}
\item {CRPS values show that post-processing globally 
improves the reliability and accuracy of 
forecasts for all watersheds and almost all forecasting lead times. This observation applies to CMC-EPS forecasts (red curve compared to blue curve) as well as general forecasts from several sources (magenta curve compared to green curve). It is therefore in our best interest to post-process ensemble maximum daily temperature forecasts produced for the 5 watersheds of Manicouagan hydropower system.}

\item {According to the CRPS, combining several forecast sources, together with statistical post-processing, reduces the average error of forecast by at least $0.5^{\circ}$C, compared to a single source of raw ensemble. This result can make all the difference when producing hydrological predictions during seasonal transitions (frost in the fall, snowmelt in the spring) since temperature forecasting plays a crucial role at these times of the year.}

\item {Surprisingly, for Toulnustouc watershed, 1 and 2-days ahead forecasts of the post-treated large ensemble obtain CRPS values (green curve) that are greater than those of the corresponding raw ensemble (magenta curve). For the 1 day ahead prediction its performance is even worse than the single source CMC-EPS raw ensemble forecasts (red curve). 
This deterioration might stem from a problem of non-homogeneity in the dataset: the model parameters fitted on the learning sample may be not appropriate for the validation one.
For instance, in Quebec, 2014 was very cold in winter compared to 2013.}

\item {Our approach obtains slightly smaller CRPS values than that of the standard EMOS method for 4 watersheds: Manic~5, Petit Lac Manic, Toulnustouc and Manic~3 (blue curve compared to yellow curve).}

\end{itemize}

\begin{figure}[h!]
   \begin{minipage}[b]{0.33\linewidth}
      \centering \includegraphics[scale=0.48]{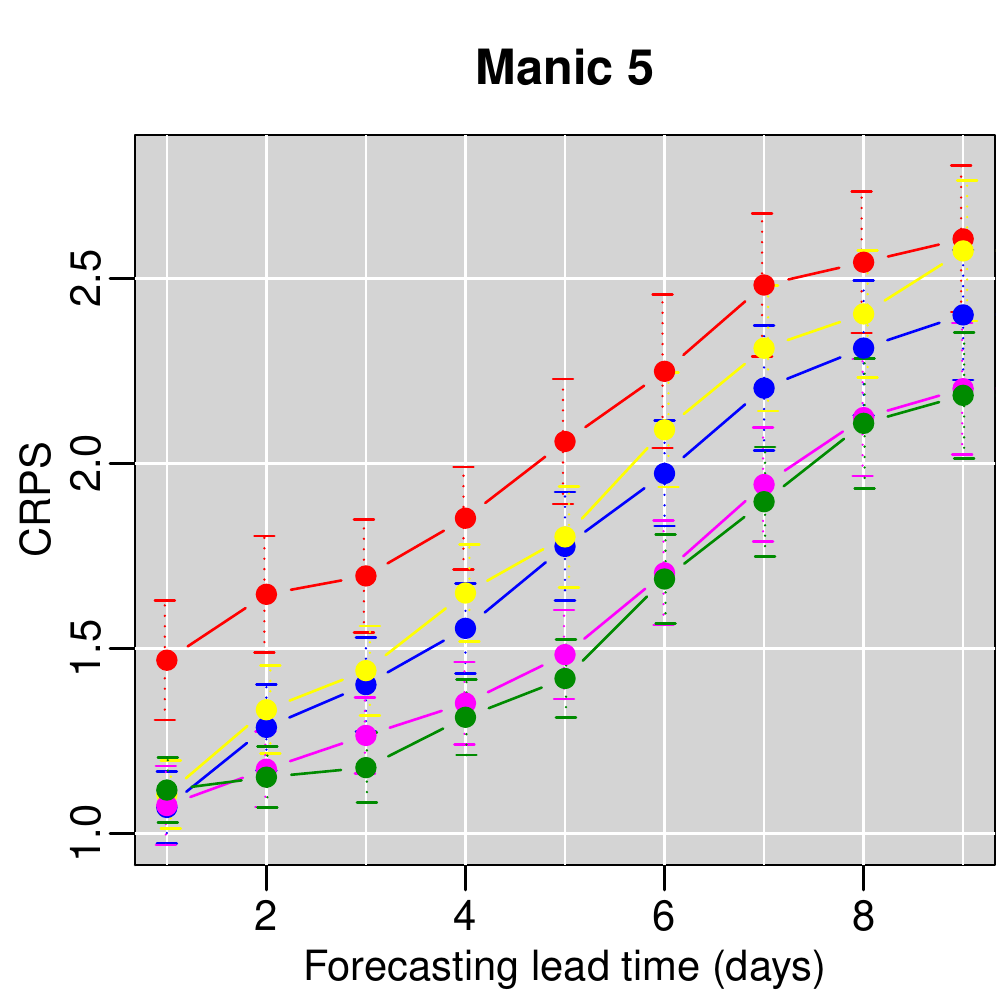}
   \end{minipage}\hfill
   \begin{minipage}[b]{0.33\linewidth}   
      \centering \includegraphics[scale=0.48]{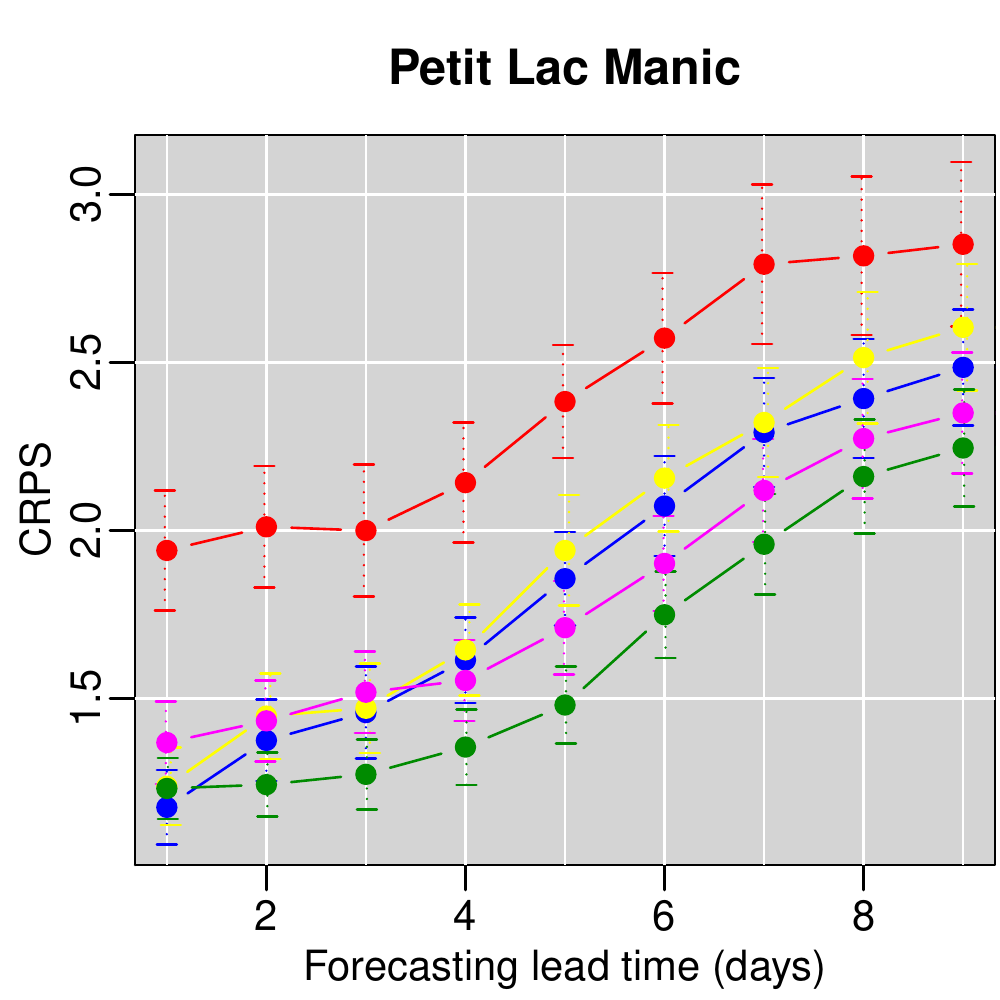}
   \end{minipage}
   \begin{minipage}[b]{0.33\linewidth}
      \centering \includegraphics[scale=0.48]{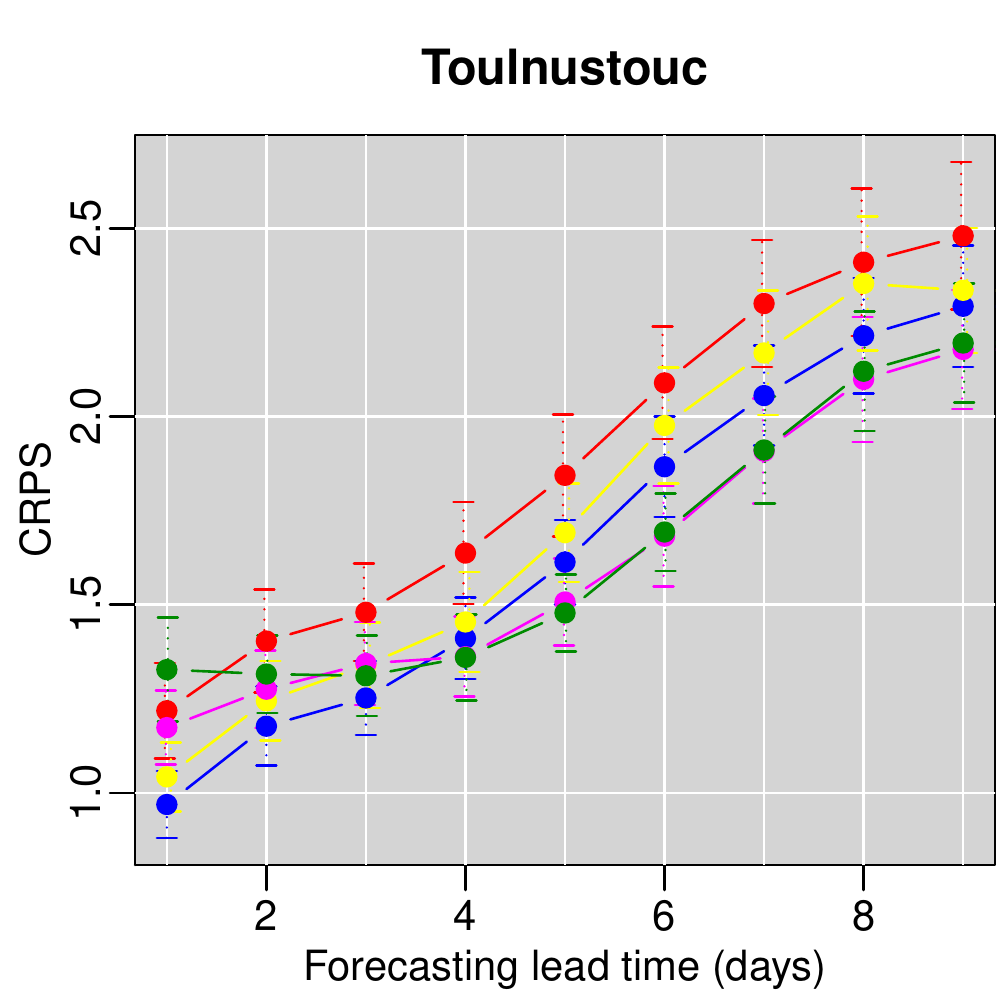}
   \end{minipage}\hfill
   \begin{minipage}[b]{0.33\linewidth}   
      \centering \includegraphics[scale=0.48]{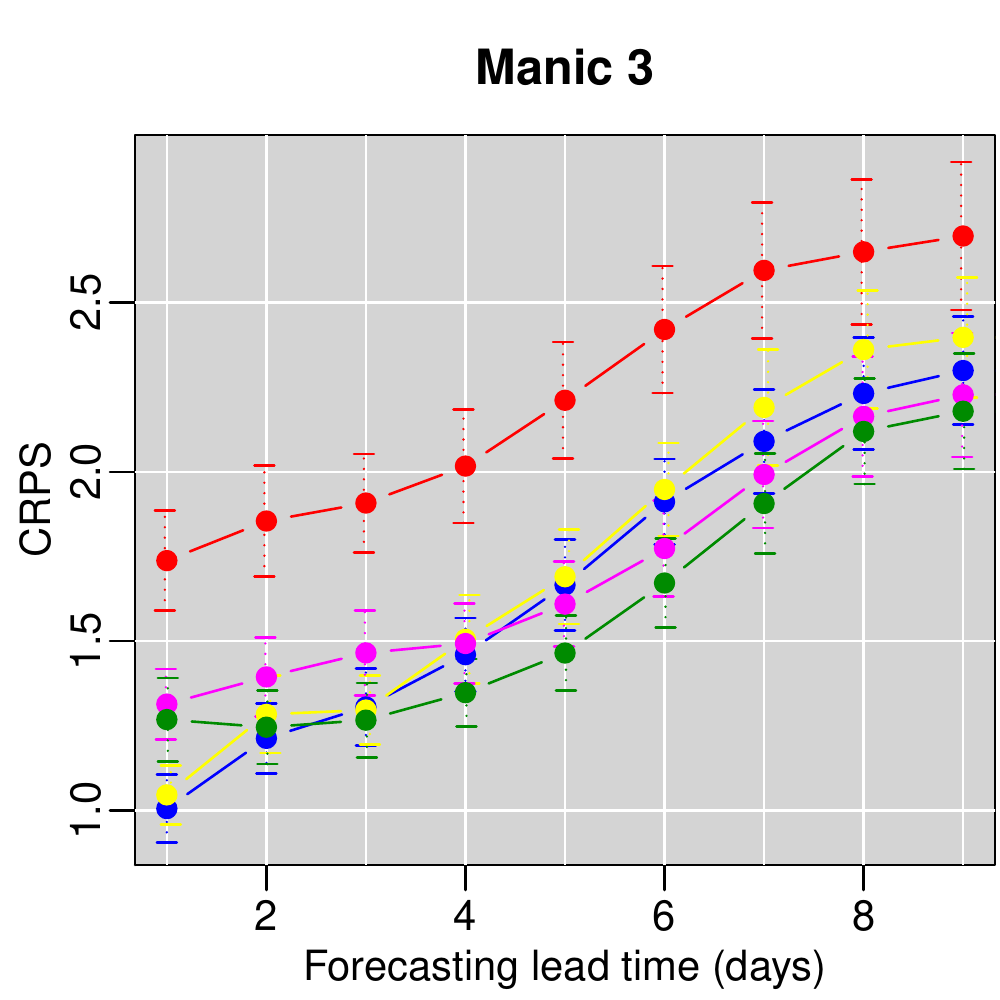}
   \end{minipage}
  \begin{minipage}[b]{0.33\linewidth}
      \centering \includegraphics[scale=0.48]{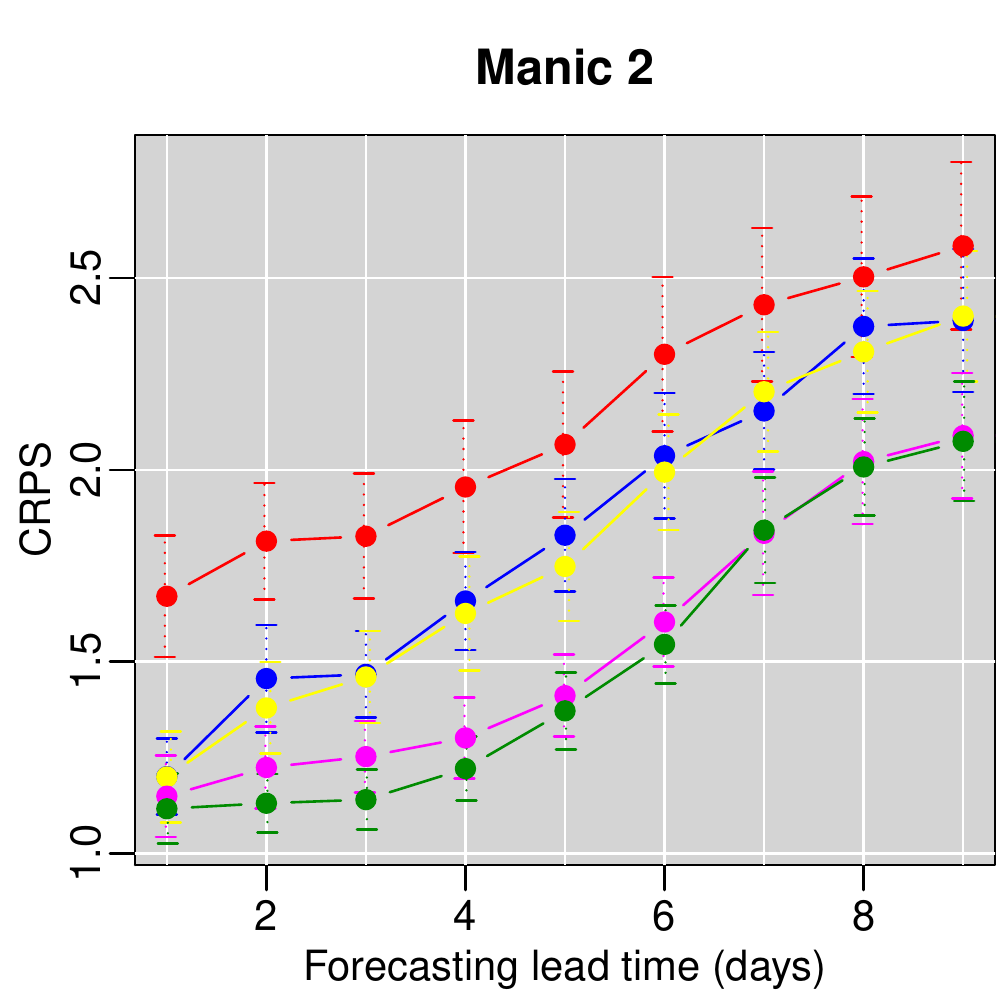}
   \end{minipage}\hfill
   \begin{minipage}[b]{0.33\linewidth}   
      \centering 
      \includegraphics[scale=0.53]{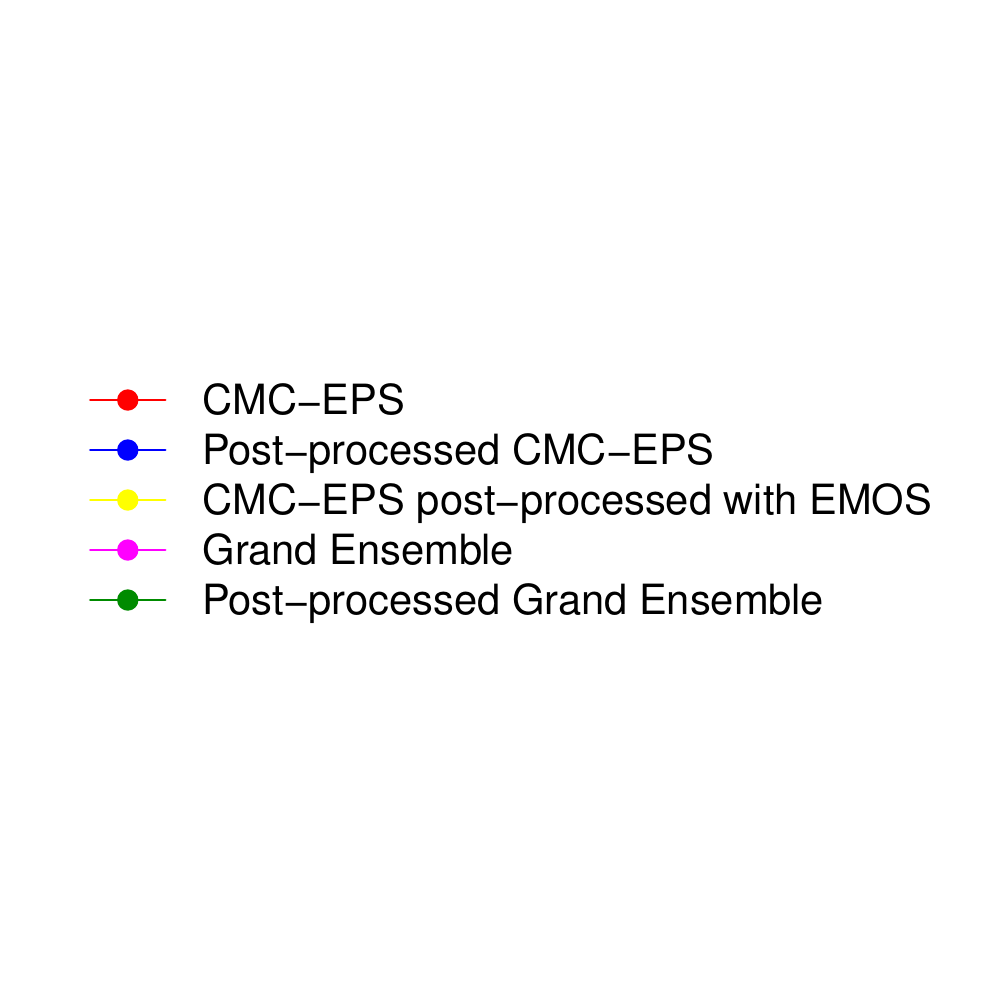}
   \end{minipage}

     \caption{CRPS (in $^{\circ}$C) for maximum temperature forecasts for the five catchments for all 9 lead times with associated bootstrap intervals. The parameters have been estimated using year 2013 and verification has been performed for year 2014. }  
     \label{Fig-CRPS5bassins}
\end{figure}

Finally, Figure~\ref{Fig-Poids} shows that observations made on 
the relative contributions of forecast sources for Manic~2 extend to the four other watersheds. This figure presents the relative contribution of each forecast source $e$, given by $ K_e\times contrib_e$, averaged over the $9$ forecast lead times. If the information provided by each ensemble were the same, the combined relative contribution of CMC-EPS and NCEP-GEFS forecasts would have been approximately $0.44$. But rather it reaches at most $0.35$, which means that the members of the CMC-EPS and NCEP-GEFS ensembles are under-weighted for the benefit of the ECMWF-EPS members.

\begin{figure}[!h]
\centering
 \includegraphics[scale=0.7]{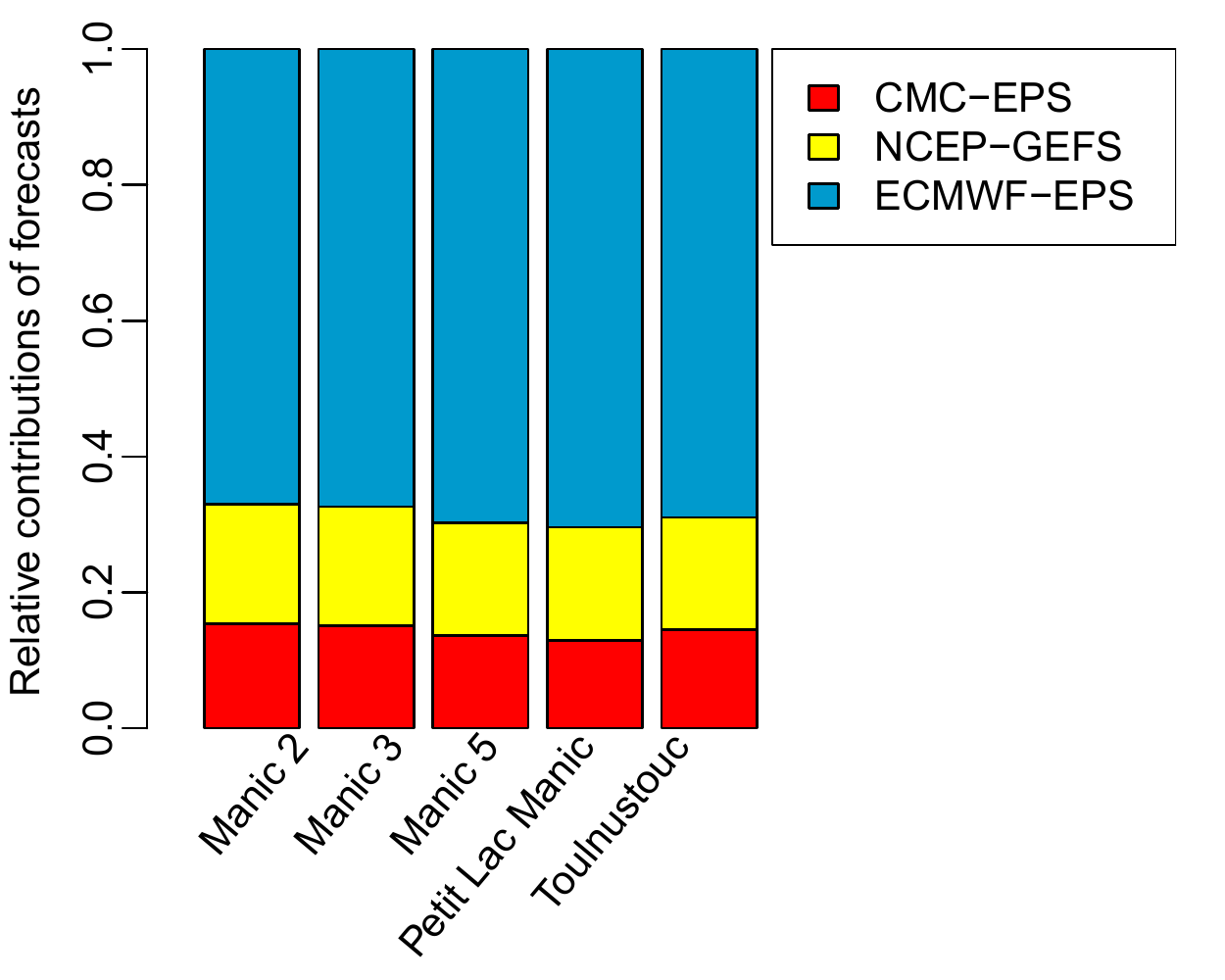} \caption{
Relative contributions of forecast sources according to the two-variable latent model averaged over the nine forecast lead times.}
 \label{Fig-Poids}
\end{figure}

\subsection{Precipitation forecasts}

\subsubsection*{Results on Manic 2 watershed}

Returning to Manic 2 watershed as an illustrative example, we now test our post-processing method for ensemble precipitation forecasts. As for temperatures, year 2013 was employed to estimate the parameters of the post-processing Tobit model, and daily forecasts were produced for 2014. 

The parameter of transformation, $\gamma$, is  estimated beforehand, the value obtained is $0.43$. 
Figure~\ref{Fig-ExManic2PrecipPrev} shows the resulting post-processed ensemble precipitation forecasts for our illustrative example of April 30th, 2014. As for temperatures, the 90 raw scenarios have been treated independently, lead time by lead time, the post-processed outputs being re-ordered using Ensemble Copula Coupling (\citep{schefzik2013uncertainty}). The color code used in this figure is the same as the one in Fig~\ref{Fig-ExManic2}.

\begin{figure}[!h]
\centering
 \includegraphics[scale=0.7]{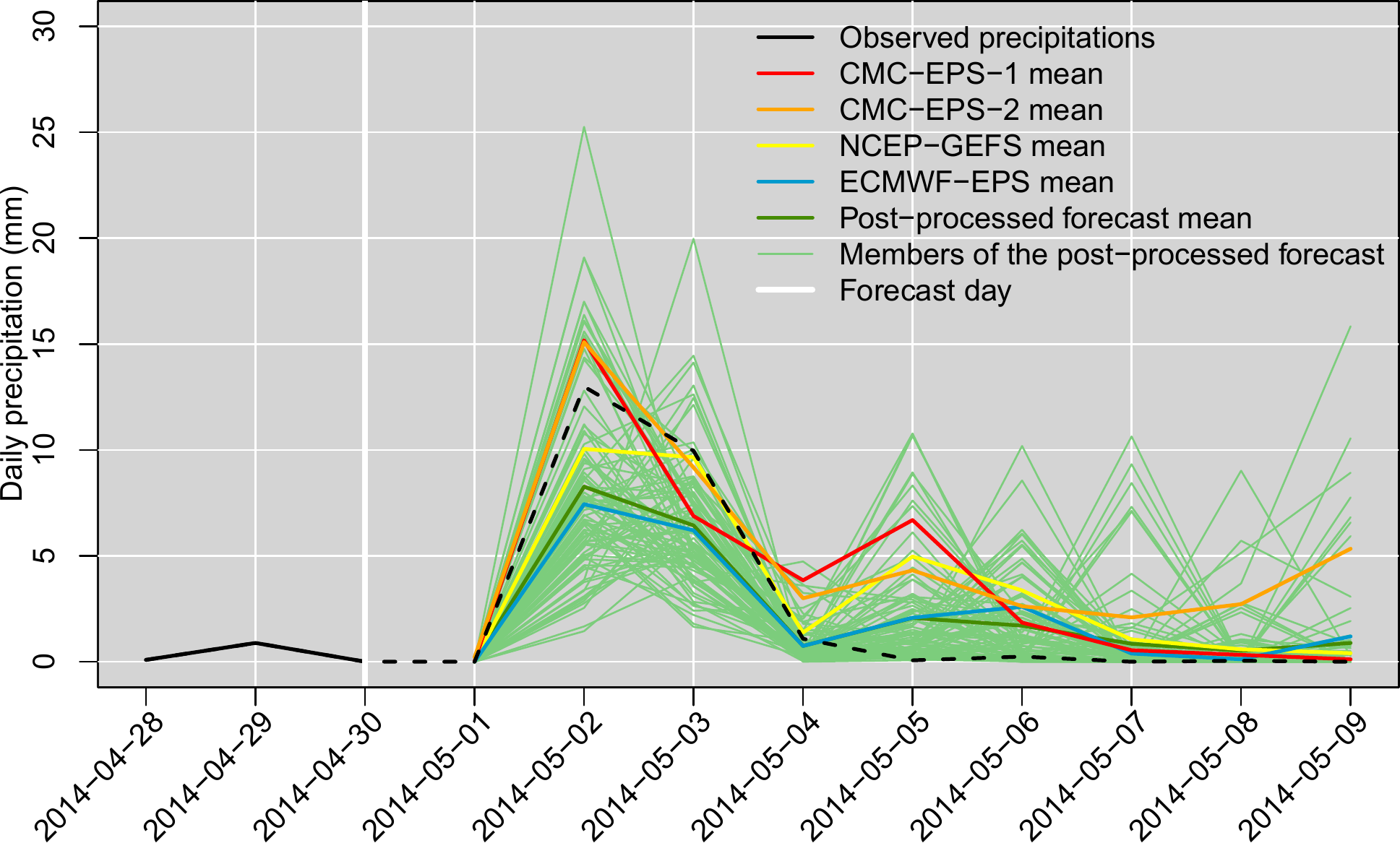}
 \caption{ 
 Example of daily precipitation forecasts for the Manic~2 watershed using the post-treatment method for NCEP-GEFS, CMC-EPS (odd~: 1 and even~: 2) and ECMWF-EPS ensembles.  Precipitation to be forecast is in dashed black.}
 \label{Fig-ExManic2PrecipPrev}
\end{figure}

Based upon visual examination of rank histograms, the post-processing method seems to lead to a better reliability compared to the original overall forecasts for the Manic 2 watershed. This gain is particularly marked for the shortest forecast deadlines for which under-dispersion and bias are more pronounced. As an illustration, the rank histograms for 3-days lead time ensemble forecasts have been presented in Fig~\ref{Fig-HistRgM2VLCens}. We see that raw ensembles (Fig~\ref{Fig-HistRgM2VLCens}--(1) to Fig~\ref{Fig-HistRgM2VLCens}--(4)) are clearly biased and underdispersive whereas post-processed ensembles show graphically good calibration (Fig~\ref{Fig-HistRgM2VLCens}--(5) and Fig~\ref{Fig-HistRgM2VLCens}--(6)). The rank histograms corresponding to the post-processed forecasts are indeed flatter than those associated with the raw ensemble forecasts. 

This gain in reliability is not reflected in the analysis of CRPS values. Figure~\ref{Fig-CRPSPrecip} shows the performance of the forecasts as measured by this scoring rule for each lead times. It is seen, for Manic 2 watershed, that post-processing does not improve significantly the raw ensemble precipitation forecasts. However, adding more ensembles is improving forecasts. In fact, CRPS values for the grand ensemble are systematically smaller than that of the CMC-EPS members.

\begin{figure}[!h]
\centering
 \includegraphics[scale=0.7,page=1]{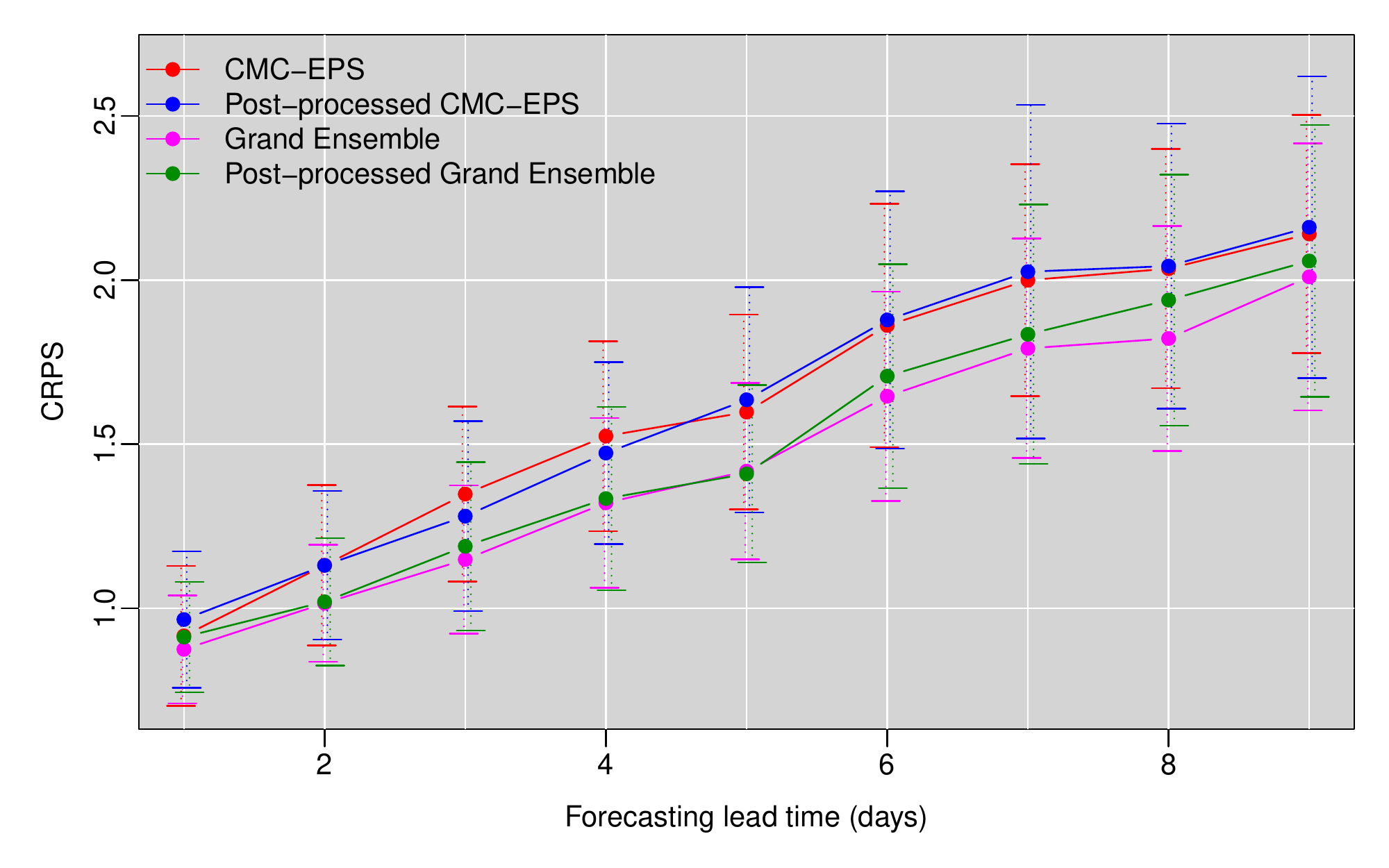}
 \caption{
CRPS (in mm) for 2014 post-processed ensemble precipitation forecasts for the Manic~2 watershed as a function of the lead times and with associated bootstrap intervals. Year 2013 is used to estimate the parameters.}
 \label{Fig-CRPSPrecip}
\end{figure}

\subsubsection*{Results on other watersheds}
For the other watersheds, the results obtained are much better: forecasts are generally slightly improved by the application of our post-treatment, particularly for small lead times, 1 up to 4 days (not shown). This is even more apparent when we consider cumulative precipitations forecasts, such as 
illustrated in Fig~\ref{Fig-cumCRPSPrecip}. In a hydrological forecasting perspective, cumulative precipitations are of particular importance since crucial decisions are taken based upon incoming inflow volumes.
\begin{figure}[!h]
\centering
 \includegraphics[scale=0.7]{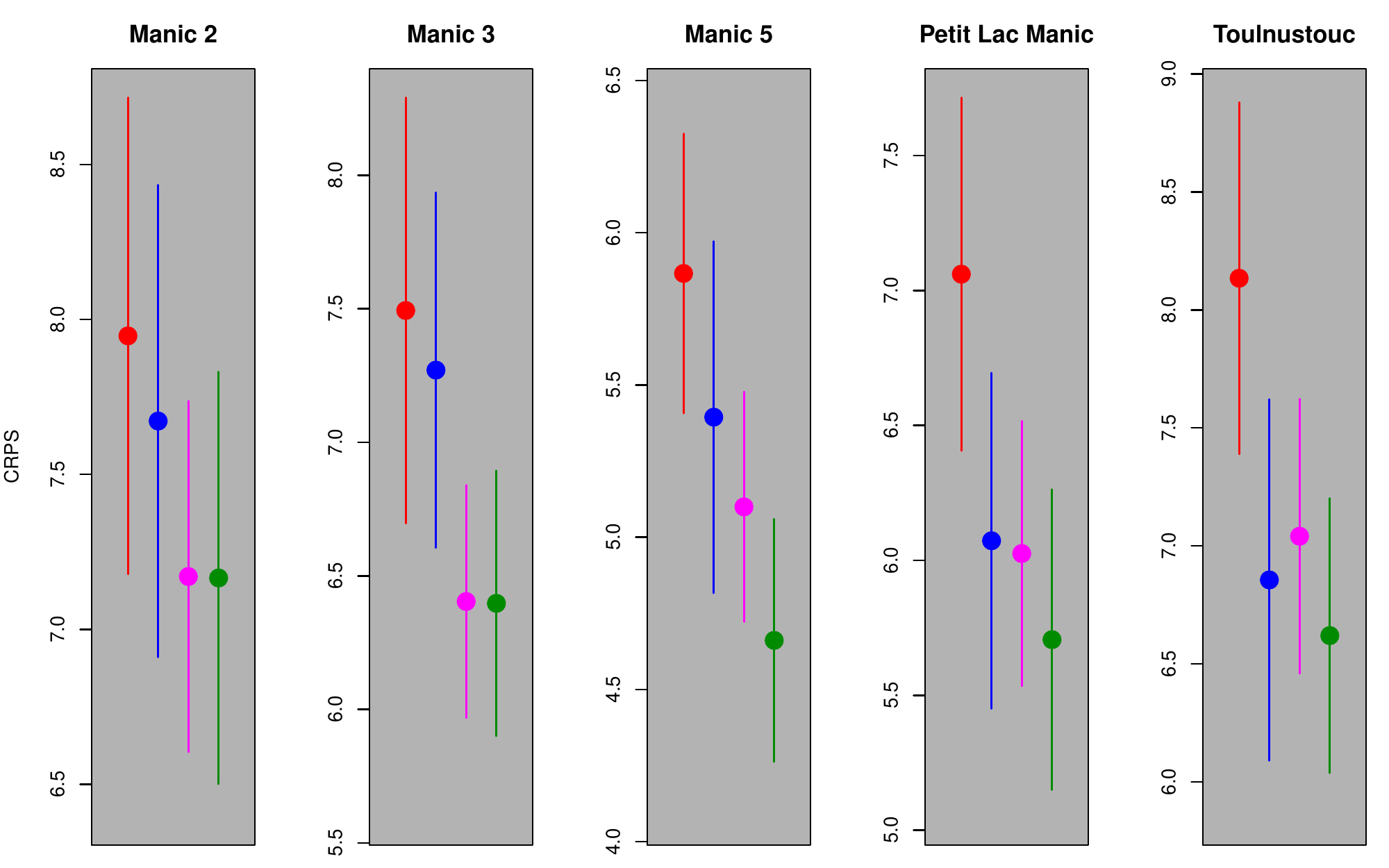}
  \includegraphics[scale=0.7]{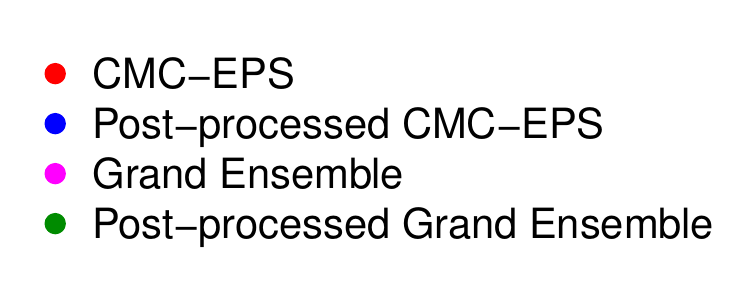}\\
 \caption{
CRPS (in mm) for 2014 post-processed 9-days ensemble cumulative precipitation forecasts and associated bootstrap intervals. Year 2013 is used to estimate the parameters.}
 \label{Fig-cumCRPSPrecip}
\end{figure}
CMC-EPS 9-days precipitation accumulation forecasts are improved by post-processing regardless of the watershed.
This is an argument in favour of a good reconstruction of the temporal dependency 
by the ECC-Q method, which has been used to get cumulative forecasts from 
the daily forecasts originally issued by our post-processing method.

The probabilistic forecasts produced using the raw large ensemble are always better than those of the CMC-EPS.
On the other hand, the forecasts for the grand ensemble totals are only improved 
for 3 out of 5 watersheds. 
This can be explained by the fact that 
the forecasts produced 
using the raw large ensemble 
already shows a good performance on those watersheds.

\section{Conclusion and discussion}
\label{sec:Conclusion-and-discussion}
Based on the concept of exchangeability, a desirable property of relabeling invariance for a meteorological ensemble, we have developed in Sections ~\ref{sec:Gaussian-case} and \ref{sec:Tobit-case} a constructive framework for  post-processing  members of multiple ensembles. This theoretically-justified mixed-effect structure gives an attractive physical interpretation to the latent variables underpinning the statistical model: they are the simple essential traits that the sophisticated numerical code mimicking the earth system model retains from the many simulations of the future weather that have been launched to give birth to the ensemble members. The model we propose also allows for the parsimonious integration of several sources of information: it can combine several ensembles produced by different meteorological centers and, eventually, deterministic weather forecasts resulting from meteorologists expertise.

We put this theoretical framework into operation for post-processing temperature forecasts, that are known for their Gaussian-like behavior as well as for precipitation forecasts, a rather more delicate statistical challenge. Due to the zero inflation of the latter distribution, we have made recourse in this case to an extension by truncation and power transform of the normal distribution. Inference methods rely on the EM and SEM algorithms to estimate the model parameters that maximize likelihood. We checked that the inference algorithms perform well on artificial data before applying the resulting post-processing methods to temperatures and precipitations data provided by Hydro-Québec.\\
Not surprisingly, the CRPS is almost systematically improved 
for all case studies after performing the post-treatments. 
Our statistical treatment affects both components of this scoring rule, calibration and sharpness. Calibration is improved because our fitted model has straightened the statistical features of the ensemble. Sharpness 
should be improved because multiple sources of information have been combined.

When adapted to precipitation ensembles by including an additional layer in the multilevel model, the method is not yet entirely satisfactory: by means of a normalized transformation of precipitation, we succeed in providing  
well calibrated forecasts, but their CRPS performances are 
a little disappointing 
when considering the multi-ensemble post-processing on 2 of the 5 considered watersheds, 
because the forecasts obtained by considering the ensemble forecasts 
all together as a grand ensemble already shows a good performance on those watersheds. 
\new{However, in an operational context, obtaining all the ensemble weather forecasts forming the grand ensemble in real time can be complicated. Hence the interest of a method providing satisfying performances from a subset of the ensembles involved, as the proposed method can do.}
\new{Moreover, the use of a moving (and/or seasonal) training period, as did \cite{taillardat2016calibrated} for instance, could be a way to get enhanced performances by better taking into account non-stationarities.} 

\new{Inference could have been performed under the Bayesian paradigm: the Bayesian interpretation of predictive conditional distributions as a personal probabilistic judgment \citep{lindley2013understanding} is most appropriate for operational forecasters.
We nevertheless chose the EM approach for inference of the models with latent layers 
since Bayesian solutions often require
more computational burden.
In addition, the Bayesian approach really brings much improvement when the experts' beliefs can be encoding, via a proper elicitation phase, to take into account supplementary information not conveyed by the data, but this is a lengthy and difficult modeling task \citep{o2005research} not developed in this paper. }

\new{Although not Bayesian, the} proposed method is able to integrate expert's forecasts as a source of additional information (it suffices to consider it as a peculiar ensemble with a single member), provided some stationarity of the process of deterministic forecast expertise (the same expert or a cohesive team of forecasters) during the learning period. This happens to be the case, for example, with Hydro-Québec. This would give forecasters, wishing to keep on working with their best single estimate of the future weather, the opportunity to share their vision of the meteorologic phenomenon within a multi-ensemble aggregation process.\\

Post-processing the meteorological ensembles to get reliable probabilistic weather forecasts is only the beginning of the story: water in the river, not in the air,  is the product of interest for hydropower companies such as HydroQuebec. As a consequence, one should also account for the uncertainty stemming from the inexact code representation of the rainfall-runoff transformation. Otherwise, some underdispersion of the probabilistic waterflow forecasts may hamper their values with regards to dam operation management. As a practical consequence, an additional statistical post-processing method is be 
used on the outputs of the rainfall-runoff model \citep{courbariaux2017water}.

Moreover, we have post-processed the meteorological ensembles for a single variable of interest, a chosen location and a given lead time. Here, as usual in the ensemble community, we simply rely on empirical copulas to restore coherence in space, time and between the multiple variables. 
Making recourse to empirical copula to retrieve a joint multivariate structure does not go without an unfortunate drawback: the number of predictive simulations cannot be greater than the size of the original ensemble.
Moreover, 
had we post-processed an ensemble for maximum and minimum daily temperatures, the method could not guarantee that the predicted maximum daily temperatures will always remain higher than the corresponding minimum daily temperatures since these two quantities of interest would have been processed independently.
A natural remedy might be to keep on taking advantage of the Gaussian properties and try to model the multivariate ensemble with many lead times and several locations as a huge Gaussian process\new{, with an autoregressive effect in the latent variable related to the mean of the ensemble}. 
This is of course the unattainable Grail \new{(the inference would be much tougher)}, but looking forward to it can help patching additional salient traits of ensembles to fruitful variations of the simple models proposed in this paper.

\section{Acknowledgements}
This work was supported by \'Electricit\'e de France and by Hydro-Qu\'ebec
	[research grant number 694R] through the PhD. thesis of M. Courbariaux.
	We would like to thank Jacques Bernier, Joël Gailhard, Anne-Catherine Favre and Vincent Fortin for their unfailing help and constructive comments regarding this work.	The forecasting and development teams at EDF-DTG
	and Hydro-Qu\'ebec have provided the necessary material
	and case studies as well as many valuable advises:
	we thank in particular Fabian Tito Arandia Martinez and 
    Éric Crobeddu from Hydro-Québec, Fabien Rinaldi and Rémy Garçon from EDF-DTG.
\bibliographystyle{abbrvnat}

\appendix

\section{Details on inference in the Gaussian case}
\label{sec:appendice:inference}

\paragraph{E-step}

We need to compute the conditional distributions
$\left[Z_t|\omega_t^{-2},\mathbf{X}_t,Y_t\right]$
and $\left[\omega_t^{-2}|\mathbf{X}_t,Y_t\right]$  
\new{for each time $t$ of the training set}.
We interpret the joint distribution \new{at time $t$}
\begin{align*}
\left[Z_t,\omega_t^{2},\mathbf{X}_t,Y_t\right] & =[\mathbf{X}_t,Y_t|Z_t,\omega_t^{-2}][\omega_t^{-2},Z_t]\\
 & =[\mathbf{X}_t|Z_t,\omega_t^{-2}][Y_t|Z_t,\omega_t^{-2}][\omega_t^{-2},Z_t]
\end{align*}
as a function of $\left(\omega_t^{-2},Z_t\right),$
and try to recognize the probability distribution function (pdf) of $\left(Z_t,\omega_t^{-2}|\mathbf{X}_t,Y_t\right)$
up to a multiplicative constant, since
$\left[Z_t,\omega_t^{-2}|\mathbf{X}_t,Y_t\right]=[\omega_t^{-2},Z_t,\mathbf{X}_t,Y_t]\times\left(\frac{1}{[\mathbf{X}_t,Y_t]}\right)$).

The complete deviance (minus twice the complete loglikelihood) \new{at time $t$} can be written as a quadratic form in $\mathbf{X}_t$ and $Y_t$, up to known normalizing constants:
\begin{multline}
\sum_{e=0}^{E}\left\{ \sum_{k=1}^{K_{e}}(X_{e,k,t}-b_{e}Z_t-a_{e})^{2}\omega_t^{-2}c_{e}^{-2}-K_{e}\log\left(\omega_t^{-2}\right)-K_{e}\log\left( c_{e}^{-2}\right)\right\} \\
\qquad+Z_t^{2}\lambda^{-1}\omega_t^{-2}-\log(\lambda^{-1})-\log\left(\omega_t^{-2}\right)+2\beta\omega_t^{-2}-2(\alpha-1)\log\left(\omega_t^{-2}\right)-2\log\left(\frac{\beta^{\alpha}}{\Gamma(\alpha)}\right)\,.\label{eq:deviancemono-1}
\end{multline}

The pdf we are looking for can be further decomposed as:
\[
\left[Z_t,\omega_t^{-2}|\mathbf{X}_t,Y_t\right]=[Z_t|\omega_t^{-2},\mathbf{X}_t,Y_t][\omega_t^{-2}|\mathbf{X}_t,Y_t]\,.
\]
We now check if we can still benefit from a conjugate situation, i.e. given $\left(\mathbf{X}_t,Y_t\right)$ we would still get a normal pdf for $\left(Z_t|\omega_t^{-2},\mathbf{X}_t,Y_t\right)$ 
under the form $\mathcal{N}(m^{\prime},\lambda^{\prime}\omega_t^{2})$ 
and a gamma pdf for $\left(\omega_t^{-2}|\mathbf{X}_t,Y_t\right)$, 
$\Gamma\left(\alpha^{\prime},\beta^{\prime}\right)$.
Expressed as a function of $[Z_t,\omega_t^{-2}|\mathbf{X}_t,Y_t]$, the deviance exhibits the following shape:
\[
(Z_t-m_t^{\prime})^{2}\lambda^{\prime-1}\omega_t^{-2}-\log(\omega_t^{-2})-\log(\lambda^{\prime-1})-2(\alpha^{\prime}-1)\log(\omega_t^{-2})+2\beta_t^{\prime}\,(\omega_t^{-2})\,.
\]
We proceed by trying to identify parameters $\alpha^{\prime},\beta_t^{\prime},\lambda^{\prime},m_t^{\prime}$
in the above equation to match the deviance for their joint distribution given by Eq~\eqref{eq:deviancemono-1}.

By matching both expressions, we obtain:
\begin{equation*}
\begin{array}{lcl}
\lambda^{\prime -1}&=&\sum_{e=0}^{E}K_{e}b_{e}^{2}c_{e}^{-2}+\lambda^{-1},\\
m_t^{\prime}&=&\lambda^{\prime}\cdot\sum_{e=0}^{E}c_{e}^{-2}b_{e}K_{e}\left(\bar{X}_{e,t}-a_{e}\right)\\
\alpha^{\prime}&=&\alpha+\frac{\sum_{e=0}^{E}K_{e}}{2},\\
\beta_t^{\prime}&=&\beta+\frac{1}{2}\left\{\sum_{e=0}^{E}\sum_{k=1}^{K_{e}}c_{e}^{-2}(X_{e,k,t}-a_{e})^{2}-m_t^{\prime2}\lambda^{\prime-1}\right\},
\end{array}
\end{equation*}
where
$\bar{X}_{e,t}=\frac{1}{K_{e}}\sum_{k=1}^{K_{e}}X_{e,k,t}$. 
Therefore, we are in a conjugate situation since the conditional pdf $\left[Z_t,\omega_t^{-2}|\mathbf{X}_t,Y_t\right]$ is in the normal-gamma model as is the marginal pdf $\left[Z_t,\omega_t^{-2}\right]$.

Denoting $\phi(\cdot)$ the first derivative of function  $\log\left\{\Gamma(\cdot)\right\}$, the moments necessary for performing the E-step are:
\begin{equation*}
\begin{array}{lcl}
\mathbb{E}\left(\log(\omega_t^{-2})|\mathbf{X}_t,Y_t\right) & =&-\log(\beta_t^{\prime})+\phi(\alpha^{\prime})\,,\\
\mathbb{E}(\omega_t^{-2}|\mathbf{X}_t,Y_t)  &=&\frac{\alpha^{\prime}}{\beta_t^{\prime}}\,,\\
\mathbb{E}(Z_t^{2}\omega_t^{-2}|\mathbf{X}_t,Y_t)  &=&\lambda^{\prime}+m_t^{\prime2}\frac{\alpha^{\prime}}{\beta_t^{\prime}}\,,\\
\mathbb{E}(Z_t \omega_t^{-2}|\mathbf{X}_t,Y_t) & =&m_t^{\prime}\frac{\alpha^{\prime}}{\beta_t^{\prime}}\,.
\end{array}
\end{equation*}

\paragraph{M-step}
We write the complete deviance $D(\btheta)=D(\alpha,\beta,\lambda,\ba,\bb,\bc)$, denoting $n$ the number of records in the data set (each of them indexed by $t$):
\begin{align*}
\begin{aligned}D(\btheta) & =
\sum_{t=1}^{n} \Bigg\{ Z_{t}^{2}\lambda^{-1}\omega_{t}^{-2}-\log(\lambda^{-1})-2\alpha\log(\omega_{t}^{-2})+2\beta\omega_{t}^{-2}-2\alpha\log(\beta)\\
 & \qquad+2\log\left\{\Gamma(\alpha)\right\}\\
 & \qquad+ \sum_{e=0}^{E}\left\{ 
 \sum_{k=1}^{K_{e}}\left(X_{e,k,t}-a_{e}-b_{e}Z_{t}\right)^{2}c_{e}^{-2}\omega_{t}^{-2}-\log\left(c_{e}^{-2}\right)
 \right\} \Bigg\}+ \textrm{Cst} \,,
\end{aligned}
\end{align*}
where $\textrm{Cst}$ is a constant term with respect to the parameters to estimate.

First, the expectation of $D(\btheta)$ is computed by using the moments computed in the E-step.
Then, this expectation is differentiated with respect to the parameters to be updated.
This leads to the following explicit update formulas, the subscript $new$ indicates the new value of the parameter:

\begin{equation*}
\begin{array}{lcl}
b_{e,new} & =&\frac{\frac{D_{e}}{B}-\frac{C_{e}}{G}}{\frac{G}{B}-\frac{H}{G}-\frac{n\lambda^{\prime}}{G\alpha^{\prime}}}\quad \textrm{ for } e\in\{1,\ldots,E\},\\
a_{e,new} & =&\frac{D_{e}}{B}-b_{e,new}\frac{G}{B} \quad \textrm{ for } e\in\{0,\ldots,E\},\\
c_{e,new}^{2}&=&K_{e}b_{e,new}^{2}\lambda^{\prime}+\frac{1}{n}\sum_{t=1}^{n}\frac{\alpha^{\prime}}{\beta_{t}^{\prime}}\sum_{k=1}^{K_{e}}\left(X_{e,k,t}-a_{e,new}-b_{e,new}m_{t}^{\prime}\right)^{2}\quad \textrm{ for } e\in\{1,\ldots,E\},\\
\lambda_{new}&=&\lambda^{\prime}+\frac{\alpha^{\prime}}{n}\sum_{t=1}^{n}\frac{m_{t}^{\prime2}}{\beta_{t}^{\prime}},\\
\beta_{new} & =&\frac{n\alpha_{new}}{\alpha^{\prime}\sum_{t=1}^{n}\frac{1}{\beta_{t}^{\prime}}},
\end{array}
\end{equation*}
where $G=\sum_{t=1}^{n}\frac{m_{t}^{\prime}}{\beta_{t}^{\prime}}$,
$B=\sum_{t=1}^{n}\frac{1}{\beta_{t}^{\prime}}$, $C_{e}=\sum_{t=1}^{n}\frac{m_{t}^{\prime}\bar{X}_{e,t}}{\beta_{t}^{\prime}}$,
$D_{e}=\sum_{t=1}^{n}\frac{\bar{X}_{e,t}}{\beta_{t}^{\prime}}$ et $H=\sum_{t=1}^{n}\frac{m_{t}^{\prime2}}{\beta_{t}^{\prime}}$.
For updating $\alpha$, we use a numeric solver of the following equation:
$$\log\left(\frac{n\alpha_{new}}{\alpha^{\prime}\sum_{t=1}^{n}\frac{1}{\beta_{t}^{\prime}}}\right)-\phi(\alpha_{new}) 
=\frac{1}{n}\sum_{t=1}^{n}\log(\beta_{t}^{\prime})-\phi(\alpha^{\prime}).$$

\end{document}